\documentclass[11pt]{article}
\usepackage{geometry}
\usepackage{setspace}
\usepackage[affil-it]{authblk}
\usepackage{booktabs}
\usepackage{amsmath,amssymb,bm, cases}
\usepackage{graphicx,psfrag,epsf}
\usepackage{caption}
\usepackage{float}
\usepackage{placeins}
\usepackage{multicol,multirow,tabularx}
\usepackage{url}
\usepackage{natbib}

\usepackage{mathtools,xparse}

\DeclarePairedDelimiter{\abs}{\lvert}{\rvert}
\DeclarePairedDelimiter{\norm}{\lVert}{\rVert}

\newtheorem{theorem}{Theorem}[section]
\newtheorem{proposition}{Proposition}[section]
\newenvironment{proof}{\paragraph{Proof:}}{\hfill$\square$}

\def\R{\mathbb{R}}
\def\bx{\mathbf{x}}
\def\bX{\mathbf{X}}
\def\by{\mathbf{y}}

\def\bv{\mathbf{v}}
\def\bV{\mathbf{V}}
\def\bT{\mathbf{T}}
\def\bsigma{\bm{\sigma}}
\def\bbeta{\bm{\beta}}

\def\rhosq{\rho^2}
\def\bs{\mathbf{s}}
\def\bSigma{\bm{\Sigma}}
\def\bmu{\bm{\mu}}
\def\tlambda{\Tilde{\lambda}_3}
\def\bLambda{\mathbf{\Lambda}}
\def\bt{\mathbf{t}}
\def\diag{\text{diag}}

\makeatletter
\ams@newcommand{\iiiiint}{\DOTSI\protect\MultiIntegral{5}}
\renewcommand{\MultiIntegral}[1]{ 
  \edef\ints@c{\noexpand\intop
    \ifnum#1=\z@\noexpand\intdots@\else\noexpand\intkern@\fi
    \ifnum#1>\tw@\noexpand\intop\noexpand\intkern@\fi
    \ifnum#1>\thr@@\noexpand\intop\noexpand\intkern@\fi
    \ifnum#1>4 \noexpand\intop\noexpand\intkern@\fi 
    \noexpand\intop
    \noexpand\ilimits@
  } 
  \futurelet\@let@token\ints@a
}
\makeatletter

\providecommand{\keywords}[1]
{
 \small	
 \textbf{\textit{Keywords:}} #1
}

\title{A New Bayesian Huberised Regularisation and Beyond}

\author[1]{Sanna Soomro}
\author[1]{Keming Yu \thanks{keming.yu@brunel.ac.uk}}
\author[2]{Yan Yu}

\affil[1]{Brunel University London}
\affil[2]{University of Cincinnati}

\date{}

\begin{document}
\maketitle

\begin{abstract}
Robust regression has attracted a great amount of attention in the literature recently, particularly for taking asymmetricity into account simultaneously and for high-dimensional analysis. However, the majority of research on the topics falls in frequentist approaches, which are not capable of full probabilistic uncertainty quantification. This paper first proposes a new Huberised-type of asymmetric loss function and its corresponding probability distribution which is shown to have the scale-mixture of normals. Then we introduce a new Bayesian Huberised regularisation for robust regression. A by-product of the research is that a new Bayesian Huberised regularised quantile regression is also derived. We further present their theoretical posterior properties. The robustness and effectiveness of the proposed models are demonstrated in the simulation studies and the real data analysis.
\end{abstract}


\setstretch{1.5} 

\section{Introduction}

Robust regression methods have a wide range of applications and attracted a great amount of attention in the literature recently, particularly for taking asymmetricity into account simultaneously and for high-dimensional analysis, such as the adaptive Huber regression  (\cite{SunEtAl2020})  and asymmetric Huber loss and asymmetric Tukey’s biweight loss functions for robust regression (\cite{FuWang2021}). The Lasso (\cite{Tibshirani1996}) and the Elastic Net (\cite{ZouTrevor2005}) are some popular choices for regularising regression coefficients. The former has the ability to automatically set irrelevant coefficients to zero. The latter retains this property and the effectiveness of the ridge penalty, and it deals with highly correlated variables more  effectively. Robust regularisation methods for quantile regression provide a promising technique for variable selection and model estimation in presence of outliers or heavy-tailed errors (\cite{LiZhu2008}; \cite{WuYuFeng2009}; \cite{Belloni2011}; \cite{Su2021}).  However, the majority of research on the topics falls in frequentist approaches, which are not capable of full probabilistic uncertainty quantification. Quantile regression, particularly Bayesian quantile regression enjoys some of robustness such as median more robust than mean, but has different modelling aims from robust regression.

Exploring  unconditional Bayesian regularisation prior, such as the Bayesian lasso (\cite{Park2008}) and the Bayesian elastic net (\cite{LiLin2010}), for robust regression is not straightforward. Several issues may arise. The joint posterior may be multimodal, which slows down the convergence of the Gibbs sampler and the point estimates may be computed through multiple modes, which lead to the inaccurate estimators (\cite{KyungEtAl2010}; \cite{Park2008}). The choices of the hyperparameters in gamma priors of regularisation parameters may also have strong influences on the posterior estimates. For the former, it was firstly observed by \cite{Park2008} in the Bayesian lasso. For the latter, it is common to employ invariant prior on scale parameter (\cite{Berger1985}). \cite{CaiSun2021} address these two issues by introducing the scale parameter to the Bayesian lasso and its generalisation for quantile regression. Moreover, \cite{Kawakami2023} use the scale parameter of the hyperbolic loss function (\cite{Park2008}) to propose the Bayesian Huberised lasso, which is the robust version of Bayesian lasso. Along this line, we will propose Bayesian Huberised regularisation in this paper. 

Quantile regression introduced by \cite{KoenkerBassett1978} is a useful supplement to ordinary mean regression in statistical analysis, owing to its robustness property and its ability to offer unique insights into the relation between the response variables and the predictors that are not available in doing mean regression.  Recently, the Bayesian approaches for variable selection in quantile regression have also attracted much attention in research area (\cite{LiEtAl2010}; \cite{Alhamzawi2012}; \cite{Alhamzawai2012-vs}; \cite{Alhamzawi2013}; \cite{ChenEtAl2013-BQR}; \cite{ReichSmith2013}; \cite{Alhamzawai2016}; \cite{Alshaybawee2017}; \cite{Adlouni2018}; \cite{Alhamzawi2019}). In Bayesian quantile regression, the error distribution would usually be assumed to follow asymmetric Laplace distribution proposed by \cite{YuMoyeed2001} that guaranteed posterior consistency of Bayesian estimators (\cite{SriramEtAl2013}) and robustness (\cite{YuMoyeed2001}). Furthermore, \cite{Alhamzawi2012} adopt the inverse gamma prior density to the penalty parameters and treated its hyperparameters as unknown and estimated them along with other parameters. This allows the different regression coefficients to have different penalisation parameters, which improves the predictive accuracy. Quantile regression, particularly Bayesian quantile regression enjoys some of robustness such as median more robust than mean, but has different modelling aims from robust regression.

Therefore, this paper first proposes a new Huberised-type of asymmetric loss function and its corresponding probability distribution, which is shown to have the scale-mixture of normals. Then we introduce a new Bayesian Huberised regularisation for robust regression. Furthermore, by taking advantage of the good quantile property of this probability distribution, we develop Bayesian Huberised lasso quantile regression and Bayesian Huberised elastic net quantile regression. This results in the proposed models covering both Bayesian robust regularisation and Bayesian quantile regularisation. Besides, \cite{CaiSun2021} emphasise that the posterior impropriety does  exist in Bayesian lasso quantile regression and its generalisation when the prior on regression coefficients is independent of the scale parameter. Thus, we will discuss some properties of the Bayesian Huberised regularised quantile regression, including posterior propreity and posterior unimodality. The approximate Gibbs sampler of \cite{Kawakami2023} is adopted to enable the data-dependent estimation of the tuning robustness parameter in the fully Bayesian hierarchical model. The advantage of this sampling step is that it does not require cross validation evaluation of tuning parameters (see \cite{Alhamzawai2016} for example) nor the rejection steps, such as the inversion method and adaptive rejection sampling algorithm (see \cite{Alhamzawi2019} for example). We demonstrate the effectiveness and robustness of the Bayesian Huberised regularised quantile regression model through simulation studies following by real data analysis. 

The remainder of this paper is as follows. In Section \ref{sec:lossfunction}, we define a Huberised asymmetric loss function with its corresponding probability density function and derive a scale mixture of normal representation for Bayesian inference. Section \ref{sec:HBR} presents the Bayesian Huberised regularisation including the Bayesian Huberised lasso (\cite{Kawakami2023}) and the Bayesian Huberised elastic net. This results in a new robust Bayesian regularised quantile regression. In Section \ref{sec:simulation} and \ref{sec:realdatanalysis}, a wide range of simulation studies and three real data examples weree conducted. In Section \ref{sec:conclusion}, we draw the conclusions. 

\section{Huberised Asymmetric Loss Function}\label{sec:lossfunction}

The lasso and elastic net estimates are all regularised estimates and the differences among them are only at their penalty terms. Specifically, they are all solutions to the following form of minimization problem for regularised quantile regression
\begin{align}
    \min_{\bbeta} \sum^n_{i=1} \rho_\tau (y_i-\bx_i\bbeta) + \lambda_1g_1(\bbeta) + \lambda_2g_2(\bbeta)\,,
    \label{eq:minimisa}
\end{align}
for some $\lambda_1,\lambda_2\geq 0$, penalty functions $g_1(\cdot)$ and $g_2(\cdot)$,   $\rho_\tau(x)=x(\tau-I(X<0))$ is the check loss function and $I(\cdot)$ is the indicator function. The lasso corresponds to  $\lambda_1>0$, $\lambda_2=0$, $g_1(\bbeta)=\norm{\bbeta}$ and $g_2(\bbeta)=0$. The elastic net corresponds to $\lambda_1=\lambda_3,\lambda_2=\lambda_4>0$, $g_1(\bbeta)=\norm{\bbeta}$ and $g_2(\bbeta)=\norm{\bbeta}^2_2$.

Letting $\tau=0.5$, the first term of Equation (\ref{eq:minimisa}) reduces to $\sum^n_{i=1}\abs{y_i-\bx_i\bbeta}$ and the corresponding method is called the least absolute deviation (LAD) regression, which is known to be robust against outliers  in response variables. However, the LAD regression might  underestimate regression coefficients for non-outlying observations. To remedy this problem, the Huber loss function is used  and it is defined as 
\begin{equation}
    L^{Huber}_{\delta}(x)=
     \begin{cases}
      \frac{1}{2} x^2, & |x|\leq \delta\,,  \\
      \delta (|x|-\delta/2), & |x|>\delta\,,
   \end{cases}
   \label{eq:huber}
\end{equation}
where $\delta>0$ is a robustness parameter and it is practically set as $\delta=1.345$ (\cite{Huber1964}). The behaviour of this loss function is such that it is quadratic for small values of $x$ and becomes linear when $\epsilon$ exceeds $\delta$ in magnitude.  

Clearly, the Huber loss function has non-differentiable points and it has limited scope in applications. \cite{LiEtAl2020} propose two generalised Huber loss functions, which are Soft Huber and Nonconvex Huber. They are attractive alternatives to the Huber loss function because they are analogous to the pseudo Huber loss function and they have a normal scale mixture property resulting in a broader range of Bayesian applications. The Soft Huber loss function can be defined as
\begin{align}
    L^{SH}_{\zeta_1,\zeta_2}(x)=\sqrt{\zeta_1\zeta_2}\left(\sqrt{1+\frac{x^2}{\zeta_2}} -1 \right)\,,
    \label{eq:softhuber}
\end{align}
and the Nonconvex Huber loss function as
\begin{align}
    L^{NH}_{\zeta_1,\zeta_2}(x)=\sqrt{\zeta_1\zeta_2}\left(\sqrt{1+\frac{|x|}{\zeta_2}} -1 \right)\,, \label{eq:nonconvexHuber}
\end{align}
where $\zeta_1, \zeta_2 >0$ are non-negative hyperparameters. Here, the Soft Huber loss bridges the $\ell_1$ (absolute) loss and the $\ell_2$ (squared) loss. On the other hand, the Nonconvex Huber loss bridges the $\ell_{1/2}$ loss and the $\ell_1$ loss. By letting $\eta=\sqrt{\zeta_1\zeta_2}$ and $\rhosq=\sqrt{\frac{\zeta_2}{\zeta_1}}$, the Soft Huber loss function becomes the hyperbolic loss function, that is,
\begin{align}
    L^{Hyp}_{\eta,\rhosq}(x) = \sqrt{\eta\left(\eta + \frac{x^2}{\rhosq} \right)} - \eta\,,
    \label{eq:hyperbolic}
\end{align}
where $\eta>0$ is the robustness parameter and $\rhosq>0$ is a scale parameter. \cite{Park2008} used this hyperbolic loss function to formulate the Bayesian Huberised lasso, which has proven to be robust to outliers. 

When the error distribution is asymmetric or contaminated by asymmetric outliers, the estimators obtained from Equations (\ref{eq:huber}), (\ref{eq:softhuber}), (\ref{eq:nonconvexHuber}) and (\ref{eq:hyperbolic}) may result in inconsistency of predictions of a conditional mean given the regressors (\cite{FuWang2021}).

Therefore, we propose the Huberised-type asymmetric loss function by letting  $\eta=\sqrt{\zeta_1\zeta_2}$ and $\rhosq=\sqrt{\frac{\zeta_2}{\zeta_1}}$ in Equation (\ref{eq:nonconvexHuber}) and it is given by 
\begin{align*}
    L^{Asy}_{\eta,\rhosq,\tau}(x) = \sqrt{\eta\left(\eta + \frac{x}{\rhosq}\left(\tau-I(x<0)\right) \right)} - \eta\,.
\end{align*}

The corresponding density function is 
\begin{align}
    f(x|\mu,\eta,\rhosq,\tau) = \frac{\eta\tau(1-\tau)e^{\eta}}{2\rhosq(\eta+1)} \exp\left\{-\sqrt{\eta\left(\eta + \frac{x-\mu}{\rhosq}\left(\tau-I(x<0)\right) \right)} \right\}\,,
    \label{eq:pdf}
\end{align}
where $\mu\in\R$ is a location parameter. Here, $\rhosq$ acts as a scale parameter and $\eta$ acts as a shape parameter of this density function. 

The following proposition states that the parameters $\mu$ and $\tau$ in (\ref{eq:pdf}) satisfy: $\mu$ is the $\tau$th quantile of the distribution. 

\begin{proposition}\label{prop:tau}
    If a random variable $X$ follows the density function in (\ref{eq:pdf}) then we have $P(X\leq\mu)=\tau$ and $P(X>\mu)=1-\tau$.
\end{proposition}
\begin{proof}
    The proof can be found in Appendix \ref{sec:appendixA}.
\end{proof}

To observe the behaviour of the proposed loss function, we set\\ $\eta=\sqrt{\zeta_2}\left(\sqrt{\zeta_2}+\sqrt{\zeta_2+1}\right)$ and $\rhosq=\frac{\sqrt{\zeta_2}}{\sqrt{\zeta_2}+\sqrt{\zeta_2+1}}$ then we have the following limits,
\begin{align*}
    \underset{\zeta_2\rightarrow 0}{\lim}\ L^{Asy}_{\eta,\rhosq,\tau}(x) = \sqrt{x\left(\tau-I(x<0)\right) } \quad \text{and} \quad \underset{\zeta_2\rightarrow \infty}{\lim}\ L^{Asy}_{\eta,\rhosq,\tau}(x) = x\left(\tau-I(x<0)\right)\,,
\end{align*}
which suggests that the proposed loss bridges the quantile loss function. \cite{daouia2018estimation} use the quantile loss function for tail expectiles to estimate alternative measures to the value at risk and marginal expected shortfall, which are two instruments of risk protection of utmost importance in actuarial science and statistical finance. \cite{ehm2016quantiles} show that any scoring function that is consistent for a quantile or an expectile functional can be represented as a mixture of elementary or extremal scoring functions that form a linearly parameterised family. However, in this paper, we  show a totally new way to achieve it, and our proposed loss is a novel representative of asymmetric least squares (\cite{daouia2019extremiles}). Figure \ref{fig:ANH}  illustrates the asymmetric shape behaviour for five different values of $\tau$ ($0.1,0.25,0.5,0.75,0.9$). From the figure, $L^{Asy}_{\eta,\rhosq,\tau}(x) $ approaches the square root of the quantile loss function, as $\eta \xrightarrow{} 0 $, and $L^{Asy}_{\eta,\rhosq,\tau}(x)$ approaches the quantile loss function, as $\eta\xrightarrow{} \infty$. 

\begin{figure}[!ht]
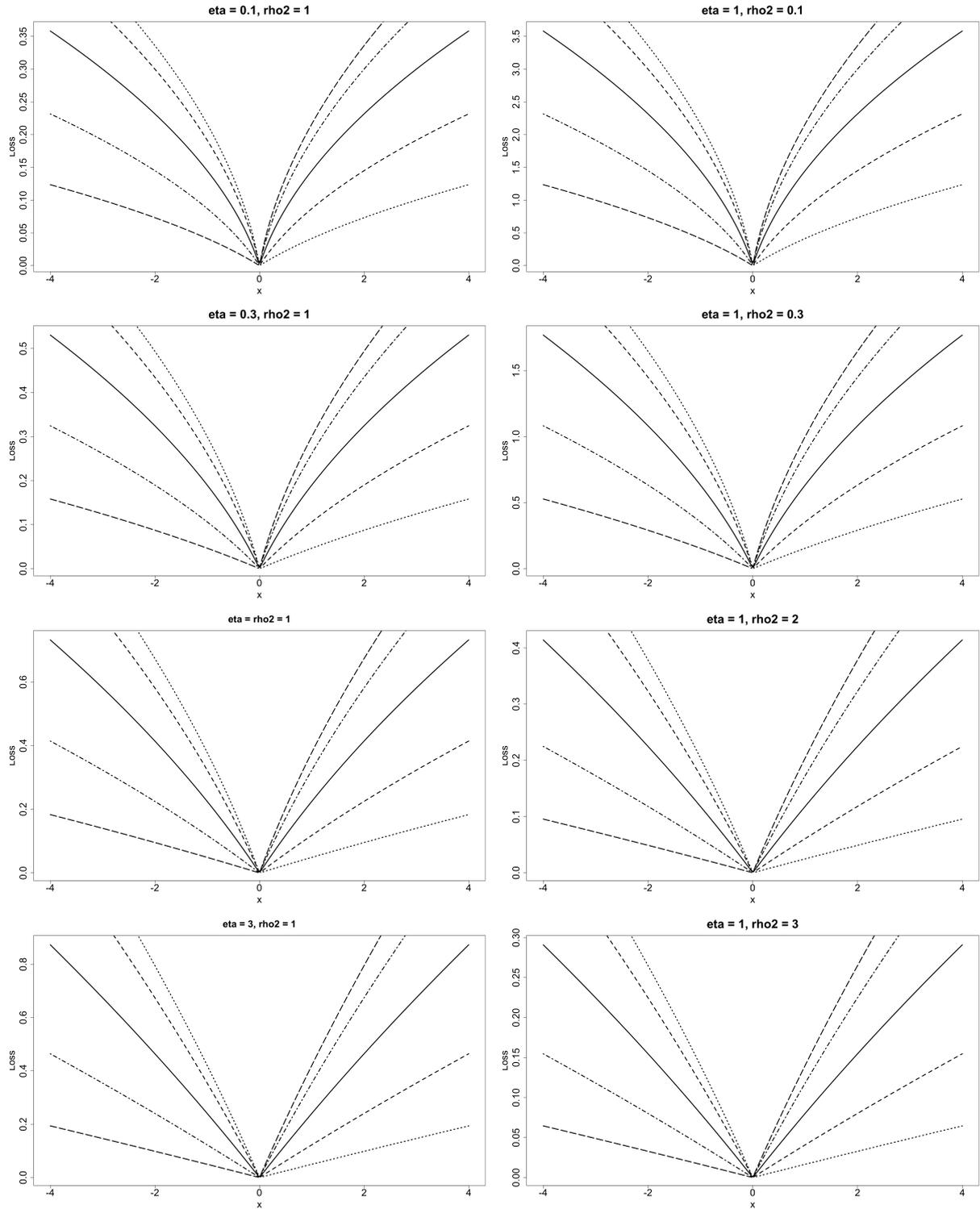

 \includegraphics[width = 0.5\textwidth]{eta01.pdf}
 \includegraphics[width = 0.5\textwidth]{rho2_01.pdf}
 \includegraphics[width = 0.5\textwidth]{eta03.pdf}
 \includegraphics[width = 0.5\textwidth]{rho2_03.pdf}
 \includegraphics[width = 0.5\textwidth]{eta1.pdf}
 \includegraphics[width = 0.5\textwidth]{rho2_2.pdf}
 \includegraphics[width = 0.5\textwidth]{eta3.pdf}
 \includegraphics[width = 0.5\textwidth]{rho2_3.pdf}
\caption{The asymmetrical behaviour of the proposed loss function for $\tau$=0.1 (short dashed), 0.25 (normal dashed), 0.5 (solid), 0.75 (short-normal dashed), and 0.9 (long dashed) for different values of $\eta$ and $\rhosq$.}
\label{fig:ANH}
\end{figure}

\cite{Kawakami2023} discussed that it is essential to choose the right value of hyperparameters of $\eta$ and $\rhosq$ where $\rhosq$ can easily be estimated by a Gibbs sampler in a Bayesian model whereas the estimation of $\eta$ is difficult. They proposed the approximate Gibbs sampler to enable the data-dependent estimation of $\eta$. This paper will also adopt their approximate Gibbs sampler. 

To fully enable the Gibbs sampling algorithm for Bayesian modelling, the density function in (\ref{eq:pdf}) has a scale mixture of normal representation with exponential and generalised inverse Gaussian densities. Suppose that a random variable $X$ has a probability density function $f(x|\theta)$ and unknown parameter $\theta$ that satisfies
\begin{equation}
 f(x|\theta)=\int \phi(x|\mu, \sigma)\,\pi(\sigma|\theta)\,d \sigma\,,
 \end{equation}
 where $\phi(\cdot)$ is the mixing distribution and $\pi(\cdot)$ is some density function that is defined on $(0,\infty)$, then $X$ or its $f(x|\theta)$ is a scale mixture of a normal distribution. It has many applications in statistics, finance and, particularly in Bayesian inference.  Probability distribution with a scale mixture of normal expression  could be grouped into two groups: symmetric probability distributions (\cite{AndrewsEtAl1974,west1987scale}) and asymmetric probability distributions (\cite{reed2009partially,da2011skew,KozumiKobayashi2011}). Therefore, the following proposition provides an alternative stochastic representation, which is a normal scale-mixture. 
 
\begin{theorem}\label{prop:scalemixture}
   If the model error $\epsilon_i=y_i-\bx_i\bbeta$ follows the density function (\ref{eq:pdf}), then we can represent $\epsilon_i$ as scale mixture of normals given by
    \begin{align}
        &f(\epsilon_i;\tau,\eta,\rhosq)\nonumber\\
        &\quad\propto\iint N \left( \epsilon_i; (1-2\tau) v_i, 4v_i\sigma_i \right) E\left(v_i; \frac{\tau(1-\tau)}{2\sigma_i}\right)GIG\left( \sigma_i; 1, \frac{\eta}{\rhosq}, \eta\rhosq  \right) dv_id\sigma_i\,, \nonumber\\
        &i=1,\ldots,n\,, \label{eq:scalemixture}
\end{align}
where $GIG(x|\nu,c,d)$ denotes the GIG distribution and its density is specified by
\begin{equation}
    f_{\text{GIG}}(x) = \frac{(c/d)^{\nu}}{2K_1(cd)} x^{\nu-1} \exp\left(-\frac{1}{2} (c^2x+d^2x^{-1}) \right)\,, \quad v>0\,, \label{eq:gig}
\end{equation}
and $K_\nu(\cdot)$ is the modified Bessel function of the second kind at index $\nu$ (\cite{BarndorffNielsen2001}).
\end{theorem}

\begin{proof}
    The proof can be found in Appendix \ref{sec:appendixA2}. 
\end{proof}

\section{Bayesian Huberised Regularised Quantile Regression Model}\label{sec:HBR}

\subsection{Bayesian Huberised Lasso Quantile Regression}

In this paper, we consider a Bayesian analogous of Huberised regularised quantile regression model. \cite{Kawakami2023} showed that the unconditional Laplace prior of $\bbeta$ (\cite{Park2008}) would lead to multimodality of a posterior density and resolved this issue by introducing $\rhosq$ as a scale parameter to formulate the Bayesian Huberised lasso, that is,
\begin{equation}
        \pi(\bbeta|\rhosq,\lambda_1) = \prod^k_{j=1} \frac{\lambda_1}{2\sqrt{\rhosq}} \exp\left\{ -\frac{\lambda_1|\bbeta_j|}{\sqrt{\rhosq}} \right\}\,.
        \label{eq:HBL}
\end{equation}
By using the scale mixture of normal representation of Laplace distribution \cite{AndrewsEtAl1974}, the Bayesian Huberised lasso can be expressed as 
\begin{align*}
    \bbeta|\bs,\rhosq \sim N(\mathbf{0},\rhosq \bLambda),\quad s_j|\lambda_1\sim Exp\left(\frac{\lambda_1^2}{2}\right)\,, \quad j=1,\ldots,k\,,
\end{align*}
where $\bs=\left(s_1,\ldots,s_k \right)^T$ and $\bLambda=\text{diag}\left( s_1,\ldots,s_k\right)$.

Therefore, with the Bayesian Huberised lasso, we present the following hierarchical model using the scale mixture of normal representation in Theorem \ref{prop:scalemixture}:
\begin{align*}
    \by|\bX,\bbeta,\bsigma,\bv &\sim N(\bX\bbeta+(1-2\tau)\bv,\bV),\\
    \sigma_i|\rhosq,\eta &\sim GIG\left(1,\frac{\eta}{\rhosq},\eta\rhosq\right)\,,\quad i=1,\ldots,n\,, \\
    v_i|\sigma_i &\sim Exp\left(\frac{\tau(1-\tau)}{2\sigma_i} \right)\,,\quad i=1,\ldots,n\,,\\
    \beta_j|s_j,\rhosq &\sim N(0,\rhosq s_j)\,, \quad j=1,\ldots,k\,,\\
    s_j|\lambda_1^2 &\sim Exp\left(\frac{\lambda^2_1}{2}\right)\,,\quad j=1,\ldots,k\,,\\
    \rhosq &\sim \pi(\rhosq) \propto \frac{1}{\rhosq}\,,\\
    \eta,\lambda_1^2 &\sim \text{Gamma}(\lambda_1^2; a,b)\text{Gamma}(\eta;c,d)\,,
\end{align*}
where $\bV = \text{diag}(4\sigma_1 v_1,\ldots,4\sigma_n v_n)$.
As a prior of $\rhosq$, we assume the improper scale invariant prior, that is proportional to $\frac{1}{\rhosq}$, but a proper inverse gamma prior can also be employed, for example. Similar to \cite{Kawakami2023} and \cite{CaiSun2021}, Proposition \ref{prop:posterior-properity_lasso} shows that using the improper prior on $\rhosq$ will lead to a proper posterior density. Baeed on this proposition, Subsection \ref{sec:multimodality} will show that the unconditional prior on $\bbeta$ can result in multimodality of the joint posterior.  We further impose a gamma prior on $\lambda_1^2$ and $\eta$. We set hyperparameters $a=b=c=d=1$ for simulation studies and real data analysis. The sensitivity analysis of hyperparameters is detailed in Subsection \ref{sec:sensitivity}. 

As for the Gibbs sampler, the full conditional distribution of $\bbeta$ is a multivariate normal distribution  and those of $\bsigma$, $\bv$, $\bs$ and $\rhosq$ are generalised inverse Gaussian distributions. The full conditional distribution of $\lambda_1^2$ is a Gamma distribution. The approximate Gibbs sampler is used for $\eta$. Appendix \ref{app:gibbs-lasso} gives the details of the full conditional posterior distributions for the Gibbs sampling algorithm. 

\begin{proposition}\label{prop:posterior-properity_lasso}
    Let $\rhosq\sim \pi(\rhosq)\propto \frac{1}{\rhosq}$ (improper scale invariant prior). For fixed $\lambda_1>0$ and $\eta>0$, the posterior distribution is proper for all $n$.
\end{proposition}
\begin{proof}
    The proof can be found in Appendix \ref{app:posterior-properity_lasso}. 
\end{proof}

\begin{proposition}\label{prop:unimodality-lasso}
    Under the conditional prior for $\bbeta$ given $\rhosq$ and fixed $\lambda_1>0$ and $\eta>0$, the joint posterior $(\bbeta,\rhosq|\by)$ is unimodal with respect to $(\bbeta,\rhosq)$.
\end{proposition}
\begin{proof}
    The proof can be found in Appendix \ref{app:unimodality-lasso}. 
\end{proof}

\subsection{Bayesian Huberised Elastic Net Quantile Regresison}

We also present the Bayesian Huberised elastic net, that is,
\begin{equation}
    \pi(\bbeta|\rhosq,\lambda_3,\lambda_4) = \prod^k_{j=1} C\left(\tlambda,\lambda_4\right) \frac{\lambda_3}{2\sqrt{\rhosq}} \exp\left\{-\frac{\lambda_3|\beta_j|}{\sqrt{\rhosq}} -\frac{\lambda_4\beta^2_j}{\rhosq} \right\}\,,
    \label{eq:HEN}
\end{equation}
where $C\left(\tlambda\,\lambda_4\right)=\Gamma^{-1}\left(\frac{1}{2},\tlambda\right) \left(\tlambda \right)^{-1/2} \exp\left\{-\tlambda \right\}$ is the normalising constant and $\tlambda=\frac{\lambda_3^2}{4\lambda_4}$. The computations of the normalising constant is detailed in Appendix B of \cite{LiEtAl2010}. Note that by letting $\rhosq=1$, Equation (\ref{eq:HEN}) reduces to the original Bayesian elastic net (\cite{LiLin2010}).

By using the scale mixture property (\cite{AndrewsEtAl1974}), the Bayesian Huberised elastic net can be expressed as a scale mixture of normal with truncated gamma density:
\begin{align*}
    \pi(\bbeta|\rhosq,\lambda_3,\lambda_4) &= \prod^k_{j=1} \int^\infty_0 \Gamma^{-1}\left(\frac{1}{2},\tlambda\right)  \sqrt{\frac{2\lambda_4 t_j}{2\pi\rhosq (t_j-1)} } \sqrt{\frac{\tlambda}{t_j}} \\
    &\quad\quad\quad\quad\times N\left(\beta_j;0,\frac{\rhosq(t_j-1)}{2\lambda_4 t_j} \right) \exp\left\{ -\tlambda t_j \right\} I(t_j>1) d\mathbf{t}\,.
\end{align*}

With the Bayesian Huberised elastic net, we have the following hierarchical model:
\begin{align*}
    \by|\bX,\bbeta,\bsigma,\bv &\sim N(\bX\bbeta+(1-2\tau)\bv,\bV)\,,\\
    \sigma_i|\rhosq,\eta &\sim GIG\left(1,\frac{\eta}{\rhosq},\eta\rhosq\right)\,,\quad i=1,\ldots,n\,, \\
    v_i|\sigma_i &\sim Exp\left(\frac{\tau(1-\tau)}{2\sigma_i} \right)\,,\quad i=1,\ldots,n\,,\\
    \beta_j|t_j,\lambda_4,\rhosq &\sim N\left(0,\frac{2\rhosq(t_j-1)}{\lambda_4 t_j}\right)\,, \quad j=1,\ldots,k\,,\\
    t_j | \tlambda &\sim \Gamma^{-1}\left(\frac{1}{2},\tlambda\right) \sqrt{\frac{\tlambda}{t_j}} \exp\left\{-\tlambda t_j \right\} I(t_j>1)\,,\quad j=1,\ldots,k\,, \\
    \rhosq &\sim \pi(\rhosq) \propto \frac{1}{\rhosq}\,,\\
    \tlambda, \lambda_4,\eta &\sim \text{Gamma}(\tlambda; a_1,b_1)\text{Gamma}(\lambda_4;a_2,b_2)\text{Gamma}(\eta;a_3,b_3)\,,
\end{align*}
where $a_1,a_2,a_3,b_1,b_2,b_3\geq 0$ are hyperparameters, they are set to $1$ for simulation studies and real data analysis and $\Gamma(\cdot,\cdot)$ is the upper incomplete gamma function. 

Appendix \ref{app:gibbs-elastic} gives the details of the full conditional posterior distributions for the Gibbs sampling algorithm. The full conditional distributions are all well-known distributions except the full conditional distributions of $\tlambda$ and $\eta$  and the Metropolis-Hasting algorithm is employed on $\tlambda$. We also present Proposition \ref{prop:unimodality-elastic} for the use of improper prior on $\rhosq$ and provide demonstration of the unconditional prior on $\bbeta$ in Subsection \ref{sec:multimodality}. 

\begin{proposition}\label{prop:posterior-properity_elastic}
    Let $\rhosq\sim \pi(\rhosq)\propto \frac{1}{\rhosq}$ (improper scale invariant prior). For fixed $\lambda_3>0$, $\lambda_4>0$ and $\eta>0$, the posterior distribution is proper for all $n$.
\end{proposition}
\begin{proof}
    The proof can be found in Appendix \ref{app:posterior-properity_elastic}. 
\end{proof}

\begin{proposition}\label{prop:unimodality-elastic}
    Under the conditional prior for $\bbeta$ given $\rhosq$ and fixed $\lambda_3>0$, $\lambda_4>0$ and $\eta>0$, the joint posterior $(\bbeta,\rhosq|\by)$ is unimodal with respect to $(\bbeta,\rhosq)$.
\end{proposition}
\begin{proof}
    The proof can be found in Appendix \ref{app:unimodality-elastic}. 
\end{proof}

\subsection{Approximate Gibbs Sampler for Estimation of $\eta$}

In this subsection, we will briefly discuss the approximate Gibbs sampler for the data-dependent estimation of $\eta$ that is proposed by \cite{Kawakami2023}. Notice that in a Bayesian Huberised regularised quantile regression model, the full conditional distribution of $\eta$ is 
\begin{align}
    \pi(\eta|\bsigma,\rhosq)\propto \frac{1}{K_1(\eta)^n} \eta^{a-1} \exp\left\{-\eta\left(\frac{1}{2}\sum^n_{i=1}\left( \frac{\sigma_i}{\rhosq} +\frac{\rhosq}{\sigma_i}\right) +b\right) \right\},
    \label{eq:eta-posterior}
\end{align}
where $a=c$ and $b=d$ in case of Bayesian Huberised lasso quantile regression and $a=a_3$ and $b=b_3$ in case of Bayesian Huberised elastic net quantile regression. Since the right side of Equation (\ref{eq:eta-posterior}) contains the modified Bessel function of the second kind, the full conditional distribution of $\eta$ does not have a conjugacy property. However, it is possible to approximate (\ref{eq:eta-posterior}) by a common probability distribution. 

For the selection of an initial value of the approximate Gibbs sampling algorithm, we need to approximate the modified Bessel function of the second kind. According to \cite{Abramowitz1965}, we have $K_\nu(x)\sim\left(\frac{1}{2} \right)\Gamma(\nu)\left(\frac{x}{2} \right)^{-\nu}$ as $x\xrightarrow{} 0$ for $\nu>0$ and $K_\nu(x)\sim \sqrt{\frac{x}{2\pi}}e^{-x}$ as $x\xrightarrow{}\infty$. \cite{Kawakami2023} stated that in either case, it would not make much difference in estimating $\eta$. So, we will focus on the latter case only for this paper. As $\eta\xrightarrow{} \infty$, we have 
\begin{align*}
    \pi(\eta|\bsigma,\rhosq) \approx \eta^{a+n/2-1}   \exp\left\{-\eta\left(\frac{1}{2}\sum^n_{i=1}\left( \frac{\sigma_i}{\rhosq} +\frac{\rhosq}{\sigma_i}\right) +b-n\right) \right\}\,,
\end{align*}
which holds the approximation $\pi(\eta|\bsigma,\rhosq)\approx \text{Gamma}\left(\eta;a+\frac{n}{2},\frac{1}{2}\sum^n_{i=1}\left( \frac{\sigma_i}{\rhosq} +\frac{\rhosq}{\sigma_i}\right) +b-n \right)$ for large $\eta$. 

The algorithm of the approximate Gibbs sampler is as follows. 

Given the current Markov chain states $(\bsigma,\rhosq)$, we set the initial value as $A=a+n/2$ and $B=\frac{1}{2}\sum^n_{i=1}\left( \frac{\sigma_i}{\rhosq} +\frac{\rhosq}{\sigma_i}\right) +b-n$. For $m=1,\ldots,M$, do the following steps
\begin{itemize}
    \item $\eta\xleftarrow{} \frac{A}{B} $;
    \item $A\xleftarrow{} a+n\eta^2\frac{\partial^2}{\partial\eta^2}\log K_1(\eta) $;
    \item $B\xleftarrow{} b+\frac{A-a}{\eta} +n\frac{\partial}{\partial\eta}\log K_1(\eta) +\frac{1}{2}\sum^n_{i=1}\left( \frac{\sigma_i}{\rhosq} +\frac{\rhosq}{\sigma_i}\right) $\,.
\end{itemize}
until $\abs{\eta/(A/B) -1 }<\varepsilon$ or in other words, the convergence of $\eta$ is met. The full derivation of the algorithm is detailed in \cite{Kawakami2023} and they also illustrated that in their simulation results, the approximation is close to the true full conditional distribution and the approximation accuracy increases as the sample size increase. For simulation studies and real data analysis, we set $M=10$ and a tolerance $\varepsilon=10^{-8}$. 

\section{Simulations}\label{sec:simulation}

\subsection{Multimodality of Joint Posteriors}\label{sec:multimodality}

\begin{figure}[ht]
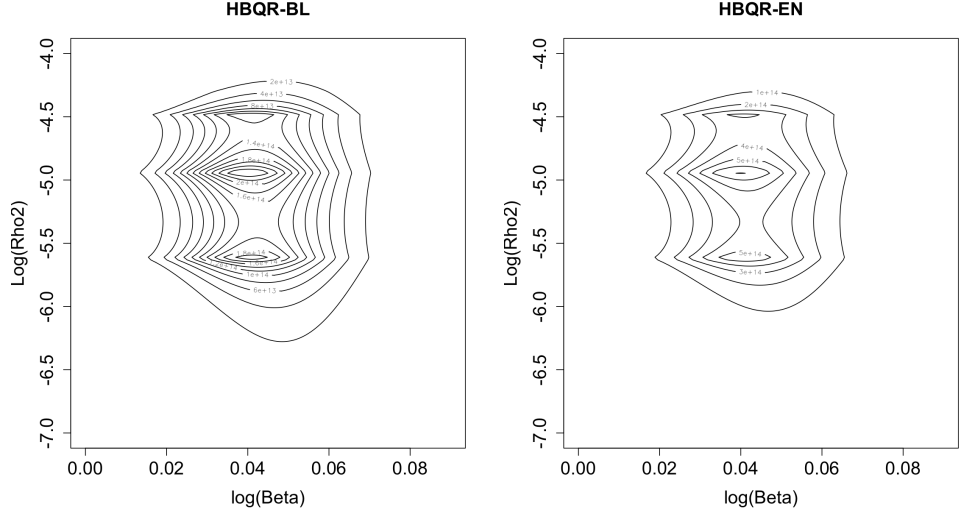

    \centering
    \includegraphics[width=0.4\textwidth]{HBQR-BL-contour.pdf}
    \includegraphics[width=0.4\textwidth]{HBQR-EN-contour.pdf}
    \caption{Contour plot of an artificially generated posterior density  of $(\log(\beta),\log(\rhosq))$ of the joint posterior density (\ref{eq:BL}) and (\ref{eq:EN}) for Bayesian Huberised lasso quantile regression and Bayesian Huberised elastic net quantile regression, respectively.  The logarithm of $\beta$ and $\rhosq$ is used for a better visibility.}
    \label{fig:contour}
\end{figure}

As related to Propositions \ref{prop:unimodality-lasso} and \ref{prop:unimodality-elastic}, we present a simple simulation to demonstrate that the unconditional prior for $\bbeta$ can result in multimodality of the joint posterior. Instead of Equations (\ref{eq:HBL}) and  (\ref{eq:HEN}), we specify the unconditional lasso prior  
\begin{align*}
    \pi(\bbeta|\lambda_1) = \prod^k_{j=1} \frac{\lambda_1}{2} \exp\left\{ -\lambda_1|\bbeta_j| \right\}\,, 
\end{align*}
and the unconditional elastic net prior
\begin{align*}
    \pi(\bbeta|\lambda_3,\lambda_4) = \prod^k_{j=1} C\left(\lambda_3,\lambda_4\right) \frac{\lambda_3}{2} \exp\left\{-\lambda_3|\beta_j| -\lambda_4\beta^2_j \right\}\,,
\end{align*}
with same improper prior $\pi(\rhosq)\propto \frac{1}{\rhosq}$. Then the joint posterior distribution of $\bbeta$ and $\rhosq$ for Bayesian Huberised lasso quantile regression is proportional to 
\begin{align}
     \pi(\bbeta,\rhosq|\by) &\propto {(\rhosq)}^{-n-1} \exp\left\{ -\lambda_1 \sum^k_{j=1} |\beta_j|\right\}  \nonumber\\
    &\quad\quad\times \prod^n_{i=1} K_0\left(\sqrt{\frac{\eta}{\rhosq}\left( \frac{|y_i-\bx_i\bbeta|+(1-2\tau)(y_i-\bx_i\bbeta)}{2}\right) } \right)\,,
    \label{eq:BL}
\end{align}
and that for Bayesian Huberised elastic net quantile regression is proportional to 
\begin{align}
    \pi(\bbeta,\rhosq|\by) &\propto{(\rhosq)}^{-n-1} \exp\left\{ - \lambda_3 \sum^k_{j=1} |\beta_j| -\lambda_4\sum^k_{j=1} \beta_j^2 \right\} \nonumber \\
    &\quad\quad\times \prod^n_{i=1} K_0\left(\sqrt{\frac{\eta}{\rhosq}\left( \frac{|y_i-\bx_i\bbeta|+(1-2\tau)(y_i-\bx_i\bbeta)}{2}\right) } \right)\,,
    \label{eq:EN}
\end{align}

In Appendices \ref{app:unimodality-lasso} and \ref{app:unimodality-elastic}, it is shown that using the conditional prior (\ref{eq:HBL}) and (\ref{eq:HEN}), respectively, lead to a unimodal posterior for any choice of $\lambda_1,\lambda_3,\lambda_4\geq 0$ and $\eta>0$ with an improper prior $\pi(\rhosq)$. On the other hand, the joint posteriors (\ref{eq:BL}) and (\ref{eq:EN}) can have more than one mode. For example, Figure \ref{fig:contour} showed the contour plots of a multimodal joint density of $\log(\beta)$ and $\log(\rhosq)$. This particular example results from considering the following data generated model,
\begin{align*}
    y_i = x_i \beta + \epsilon_i\,, \quad \epsilon_i\sim ALD(0,\sigma=0.03,\tau=0.5)\,,
\end{align*}
where $\beta=1$ and $x_i\sim N(0,1)$ for $i=1,\ldots,10$, which is similar to \cite{CaiSun2021}. Due to multimodality in the joint posterior with unconditional prior for $\bbeta$, we use the prior for $\bbeta$ conditioning on the scale parameter $\rhosq$. 

\subsection{Sensitivity analysis of hyper-parameters}\label{sec:sensitivity}

In this subsection, we test the sensitivity of hyperparameters of Gamma prior of $\eta$, $\lambda_1$, $\lambda_3$ and $\lambda_4$ on the posterior estimates for the proposed methods. We equally divide $x\in[-2,2] $ into $50$ pieces and the data are generated from 
\begin{align*}
    y_i = \mathbf{x}_i\bm{\beta} + \epsilon_i\,,\quad \epsilon_i\sim ALD(0,\sigma=0.03,\tau=0.5)\,, \quad i=1,\ldots,50\,,
\end{align*}
with $\mathbf{x}_i=\left( \left( 1+e^{-4(x_i-0.3)} \right)^{-1}, \left( 1+e^{3(x_i-0.2)} \right)^{-1}, \left( 1+e^{-4(x_i-0.7)} \right)^{-1}, \left( 1+e^{5(x_i-0.8)} \right)^{-1} \right)^T$ and $\bm{\beta}=(1,1,1,1)^T$. It indicates that the true curve is
\begin{align*}
    f(x) =  \left( 1+e^{-4(x-0.3)} \right)^{-1} + \left( 1+e^{3(x-0.2)} \right)^{-1} + \left( 1+e^{-4(x-0.7)} \right)^{-1} + \left( 1+e^{5(x-0.8)} \right)^{-1}\,.
\end{align*}
In fact, this function was utilised in \cite{JullionLambert2007} to test the sensitivity of hyperparameters of the Gamma prior on the scale component in Bayesian P-spline.

We consider the proposed models to estimate $\bm{\beta}$. Note that there are four prior hyperparameters a, b, c and d in the Bayesian Huberised lasso quantile regression and six prior hyperparamters $a_1$, $b_1$, $a_2$, $b_2$, $a_3$ and $b_3$ in the Bayesian Huberised elastic net quantile regression. We mainly set $a=b=c=d=a_1=a_2=a_3=b_1=b_2=b_3=1$ in both simulation studies and data analysis. We generate $3000$ posterior samples after discarding the first 1000 posterior samples as burn-in. Then we plot $y_i=\mathbf{x}_i\bm{\beta}$ for $i=1,\ldots,50$ in Figures \ref{fig:HBQR-BL-hyper} and \ref{fig:HBQR-BEN-hyper} for both proposed Bayesian models, where $\bm{\beta}$ is the posterior mean for the corresponding proposed model. In Figure \ref{fig:HBQR-BL-hyper}, we fixed $a=1$ with $b$ varied for the top-left plot and $b=1$ with $a$ varied for the top-right plot. In both cases, we  keep $c=d=1$ fixed. Both bottom plots of Figure \ref{fig:HBQR-BL-hyper} follows in a similar manner. As for Figure \ref{fig:HBQR-BEN-hyper}, we also fixed $a_1=1$ with $b_1$ varied for the top-left plot while keeping $a_2=b_2=a_3=b_3=1$. The rest of Figure \ref{fig:HBQR-BEN-hyper} also follows in a similar manner.   From the figures, we observe that the estimation results do not change very much for a variety selection of hyperparameters.

\begin{figure}[H]
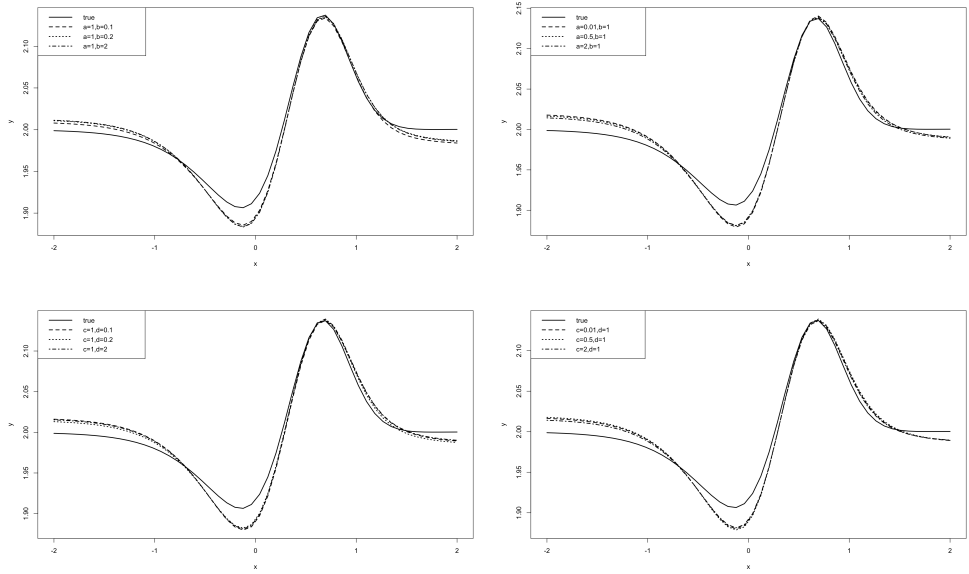

    \centering
    \includegraphics[width=0.4\textwidth]{HBQR_BL_b122.pdf}
    \includegraphics[width=0.4\textwidth]{HBQR_BL_a152.pdf}
    \hfill
    \includegraphics[width=0.4\textwidth]{HBQR_BL_d122.pdf}
    \includegraphics[width=0.4\textwidth]{HBQR_BL_c152.pdf}
    \caption{Sensitivity analysis of hyper-parameters for the Bayesian Huberised lasso quantile regression.}
    \label{fig:HBQR-BL-hyper}
\end{figure}
\FloatBarrier

\begin{figure}[H]
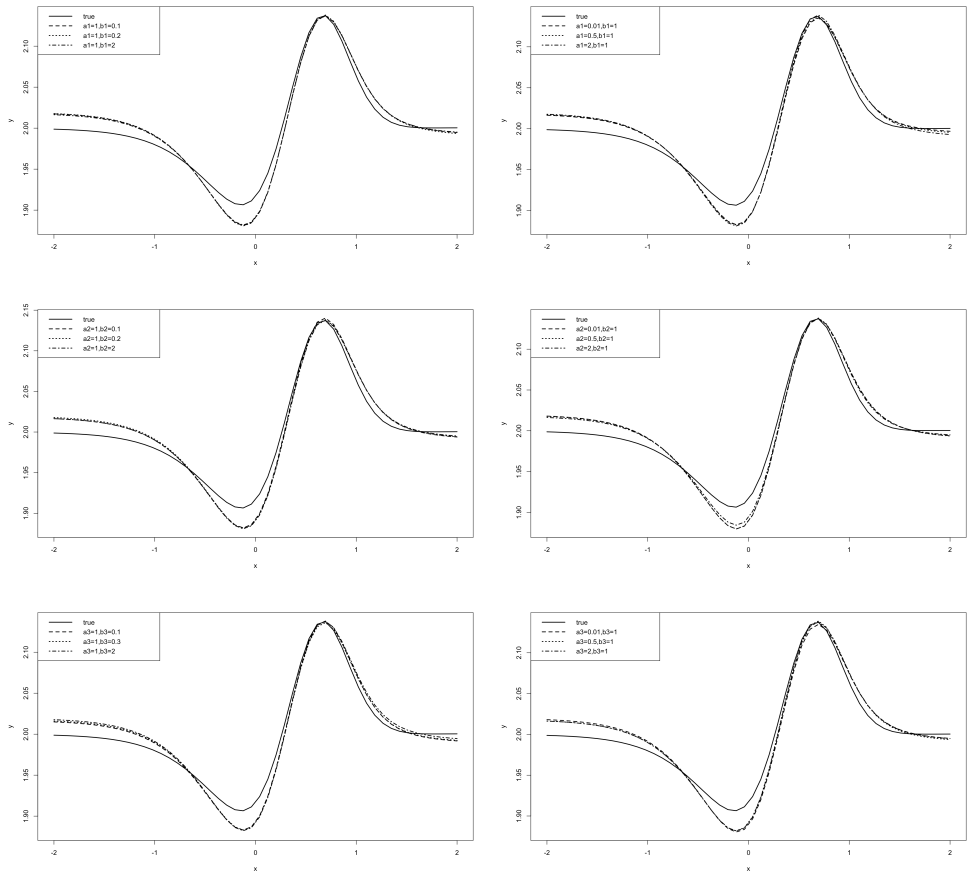

    \centering
    \includegraphics[width=0.4\textwidth]{HBQR_BEN_b1_122.pdf}
    \includegraphics[width=0.4\textwidth]{HBQR_BEN_a1_152.pdf}
    \hfill
    \includegraphics[width=0.4\textwidth]{HBQR_BEN_b2_122.pdf}
    \includegraphics[width=0.4\textwidth]{HBQR_BEN_a2_152.pdf}   
    \hfill
    \includegraphics[width=0.4\textwidth]{HBQR_BEN_b3_122.pdf}
    \includegraphics[width=0.4\textwidth]{HBQR_BEN_a3_152.pdf}
    \caption{Sensitivity analysis of hyper-parameters for the Bayesian Huberised elastic net quantile regression.}
    \label{fig:HBQR-BEN-hyper}
\end{figure}
\FloatBarrier

\subsection{Simulation Studies}

In simulation studies, we illustrate performance of the proposed methods. We compare
the point and interval estimation performance of the proposed methods with those of
some existing methods. To this end, we consider the following regression model with
$n\in \{100,200\}$, $k=20$ and $\tau\in\{0.25,0.5,0.75\}$:
\begin{align*}
    y_i = \beta_0 + \beta_1x_{i1} +\ldots +  \beta_kx_{ik} +\sigma\epsilon_i\,,\quad i=1,\ldots,n\,,
\end{align*}
where $\beta_0=1$, $\beta_1=3$, $\beta_2=0.5$,$\beta_4=\beta_{11}=1$, $\beta_7=1.5$ and the other $\beta_j$'s were set to $0$. We assume $\mathbf{y}=(y_1,\ldots,y_n)^T$ is the response vector. The predictors $\mathbf{x}_i=(x_{i1},\ldots,x_{ik})^T$ were generated from a multivariate normal distribution $N_k(0,\Sigma)$ with $\Sigma=(r^{|i-j|})_{1\leq i,j\leq k}$ for $|r|<1$. Similar to \cite{Kawakami2023} and \cite{LambertLaurent2011}, we consider the siz scenarios.

\begin{itemize}
    \item Simulation 1: Low correlation and Gaussian noise. $\bm{\epsilon}\sim N_n(0,I_n)$, $ \sigma=2$ and $r=0.5$.
    \item Simulation 2: Low correlation and large outliers. $\epsilon=W/\sqrt{var(W)}$, $\sigma=9.67$ and $r=0.5$. $W$ is ta random variable according to the contaminated density defined by $0.9\times N(0,1)+0.1\times N(0,15^2)$, where $\sqrt{var(W)}=4.83$.
    \item Simulation 3: High correlation and large outliers. $\epsilon=W/\sqrt{var(W)}$, $\sigma=9.67$ and $r=0.95$.
    \item Simulation 4: Large outliers and skew Student-t noise. $\epsilon_i \sim 0.9\times \text{Skew-}t_3(\gamma=3) + 0.1\times N(0,20^2) $, $\sigma=1$ and $r=0.5$.
    \item Simulation 5: Heavy-tailed noise. $\epsilon_i\sim \text{Cauchy}(0,1)$, $\sigma=2$ and $r=0.5$.
    \item Simulation 6: Multiple outliers.  $\epsilon_i \sim 0.8\times \text{Skew-}t_3(\gamma=3) + 0.1\times N(0,10^2) +0.1\times \text{Cauchy}(0,1)$, $\sigma=1$ and $r=0.5$.
\end{itemize}

For the simulated dataset, we applied the proposed robust methods denoted by HBQR-BL and HBQR-EN where they were employed with Bayesian Huberised Lasso and Bayesian Huberised elastic net, respectively. We also applied the existing robust methods, including Bayesian linear regression with Bayesian Huberised lasso (\cite{Kawakami2023}), and Bayesian quantile regression with original Bayesian lasso and Bsyesian elastic net (\cite{LiEtAl2010}) denoted by HBL, BQR-BL and BQR-EN, respectively. For HBL and BQR-BL, we assume $\lambda_1\sim \text{Gamma}(a=1,b=1)$ and for BQR-EN, we assume $\lambda_1\sim \text{Gamma}(a_1=1,b_1=1)$ and $\lambda_2\sim \text{Gamma}(a_2=1,b_2=1)$. For the HBQR-BL and HBQR-EN, We implement both Gibbs and Metropolis-within-Gibbs algorithms, respectively and set all the hyperparameters to 1. 

When applying the above methods, we generated 2000 posterior samples after discarding the first 500 samples as burn-in. We computed posterior median of each element of $\beta_j$'s for point estimates of $\beta_j$'s, and the performance is evaluated via root of mean squared error (RMSE) defined as \\$\left[(k+1)^{-1} \sum^k_{j=0} (\hat{\beta}_j-\beta_j^{\text{true}})^2 \right]^{1/2}$, and median of mean absolute error (MMAD) defined as \\$\text{median}\left[(k+1)^{-1} \sum^k_{j=0} \left|\hat{\beta}_j-\beta_j^{\text{true}}\right| \right]$. We also computed $95\%$ credible intervals of $\beta_j$'s, and calculated average lengths (AL) and coverage probability (CP) defined as $(k+1)^{-1} \sum^k_{j=0} |CI_j|$ and $(k+1)^{-1} \sum^k_{j=0} I(\beta_j\in CI_j)$, respectively. These values were averaged over 300 replications of simulating datasets. 

We report simulation results in Tables \ref{tab:model1}-\ref{tab:model6} and Figures \ref{fig:boxplot-rmse}-\ref{fig:boxplot-eta}.  In Simulation 1, there is no outliers in simulated datasets. In the median case ($\tau=0.5$), both HBL and BQR-BL have the smallest RMSE and MMAD for both $n=100$ and $n=200$. In the upper and lower quantile cases ($\tau=0.25,0.75$), HBQR-BL and HBQR-EN outperformed BQR-BL and BQR-EN even though they are comparable. In the presence of large outliers (Simulations 2 and 3) for $\tau=0.5$, HBQR-BL and HBL perform better for $n=100$ for Simulation 2. As the sample size increases to $n=200$, BQR-BL and HBL outperform the proposed methods. Similarly, HBL and BQR-EN have the smallest RMSE and MMAD for Simulation 3. However, for both simulations, the proposed methods perform significantly better in case of upper and lower quantiles. Particularly, in Simulations 4-6 where there are skewed \& heavy-tailed noise with large outliers, heavy-tailed noise (Cauchy distribution) and multiple outliers, respectively, the proposed methods perform significantly better than the existing robust methods in all cases of $\tau$. Observing the performance of BQR-EN, the boxplots of Simulation 5 in Figures \ref{fig:boxplot-rmse}-\ref{fig:boxplot-al} were wider compared to other methods, which suggest that this method may not produce as efficient estimates as the proposed HBQR-EN did.  All the boxplots for $\tau=0.25$ and $\tau=0.75$ can be found in Appendix \ref{app:figures}. Looking at the CP, they are reasonable in all simulations. For AL, it is evident that the proposed methods have the lowest AL in all cases while the BQR-BL has the largest AL.  To conclude these simulation studies, the proposed methods seem to perform well consistently in all scenarios. 

We also present the boxplot of posterior median of $\eta$ in Figure \ref{fig:boxplot-eta} for $\tau=0.5$. In the absence of outliers (Simulation 1), the posterior median of $\eta$ has large values. On the other hand, in the present of large outliers (Simulations 2-4,6) and in a model following a heavy-tailed noise (Simulation 5), small $\eta$ is chosen. The results for $\tau=0.25$ and $\tau=0.75$ are similar (see Appendix \ref{app:figures}). Therefore, like the HBL method (\cite{Kawakami2023}), it is evident that $\eta$ is adaptively chosen for each simulated dataset. 

\begin{figure}[H]
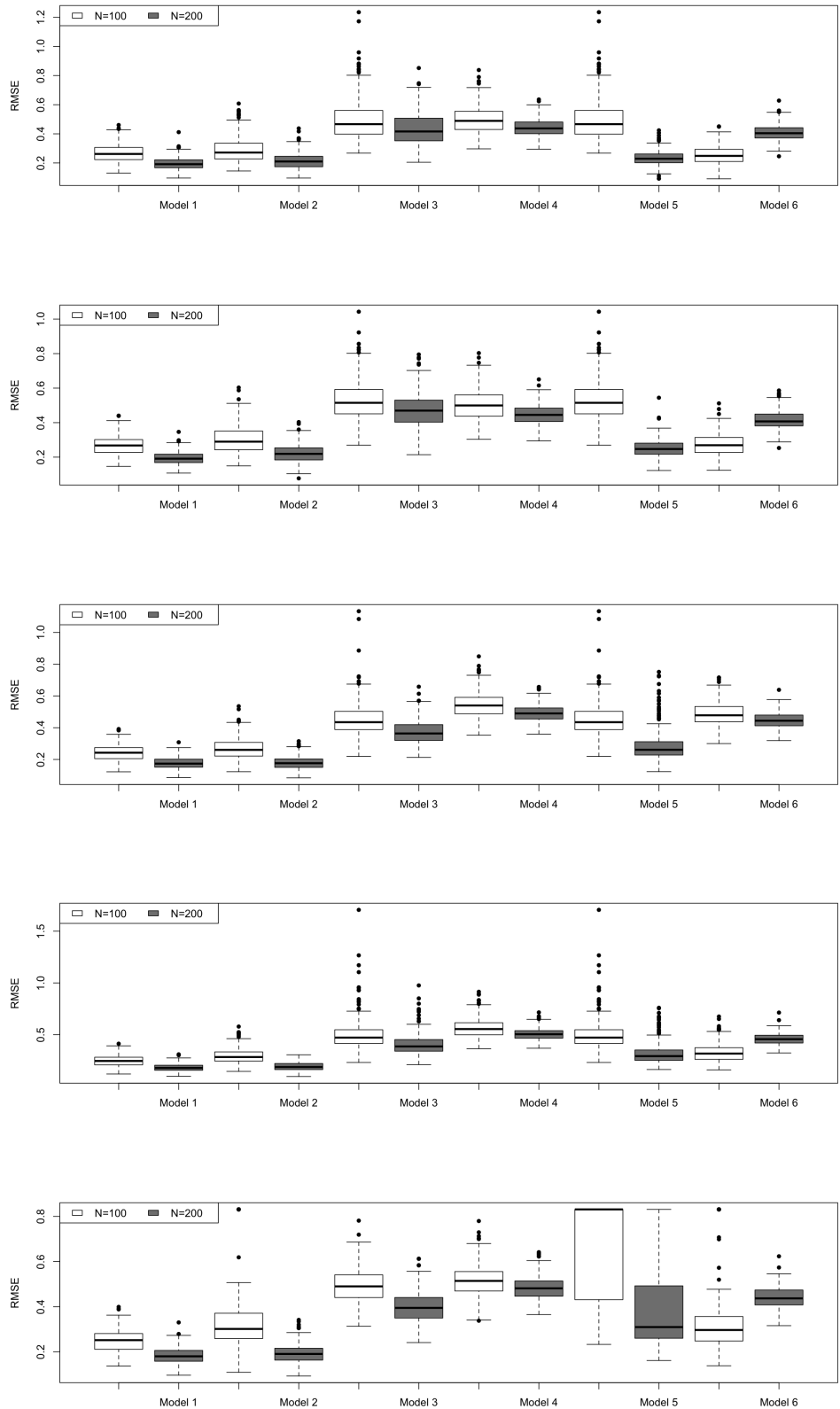

    \centering
    \includegraphics[width=0.8\textwidth]{HBQR_BL_RMSE_tau5.pdf}
    \includegraphics[width=0.8\textwidth]{HBQR_BEN_RMSE_tau5.pdf}
    \includegraphics[width=0.8\textwidth]{HBL_RMSE_tau5.pdf}
    \includegraphics[width=0.8\textwidth]{BQR_BL_RMSE_tau5.pdf}
    \includegraphics[width=0.8\textwidth]{BQR_BEN_RMSE_tau5.pdf}
    \caption{Boxplots of RMSE based on 300 replications in six simulation scenarios for HBQR-BL, HBQR-EN, HBL, BQR-BL and BQR-EN in this order ($\tau=0.5$).}
    \label{fig:boxplot-rmse}
\end{figure}
\FloatBarrier

\begin{figure}[H]
    \centering
    \includegraphics[width=0.8\textwidth]{HBQR_BL_MMAD_tau5.pdf}
    \includegraphics[width=0.8\textwidth]{HBQR_BEN_MMAD_tau5.pdf}
    \includegraphics[width=0.8\textwidth]{HBL_MMAD_tau5.pdf}
    \includegraphics[width=0.8\textwidth]{BQR_BL_MMAD_tau5.pdf}
    \includegraphics[width=0.8\textwidth]{BQR_BEN_MMAD_tau5.pdf}
    \caption{Boxplots of MMAD based on 300 replications in six simulation scenarios for HBQR-BL, HBQR-EN, HBL, BQR-BL and BQR-EN in this order ($\tau=0.5$).}
    \label{fig:boxplot-mmad}
\end{figure}
\FloatBarrier

\begin{figure}[H]
    \centering
    \includegraphics[width=0.8\textwidth]{HBQR_BL_AL_tau5.pdf}
    \includegraphics[width=0.8\textwidth]{HBQR_BEN_AL_tau5.pdf}
    \includegraphics[width=0.8\textwidth]{HBL_AL_tau5.pdf}
    \includegraphics[width=0.8\textwidth]{BQR_BL_AL_tau5.pdf}
    \includegraphics[width=0.8\textwidth]{BQR_BEN_AL_tau5.pdf}
    \caption{Boxplots of AL based on 300 replications in six simulation scenarios for HBQR-BL, HBQR-EN, HBL, BQR-BL and BQR-EN in this order ($\tau=0.5$).}
    \label{fig:boxplot-al}
\end{figure}
\FloatBarrier

\begin{figure}[H]
    \centering
    \includegraphics[width=\textwidth]{HBQR_BL_Eta_tau5.pdf}
    \includegraphics[width=\textwidth]{HBQR_BEN_Eta_tau5.pdf}
    \caption{Boxplots of posterior median of $\eta$ based on 300 replications in six simulation scenarios for HBQR-BL (top) and HBQR-EN (bottom) ($\tau=0.5$).}
    \label{fig:boxplot-eta}
\end{figure}
\FloatBarrier

\begin{table}[H]
    \centering
    \begin{tabular}{|c|c|c|c|c|c|}
    \hline
         & Methods & RMSE & MMAD & AL & CP  \\
         \hline
        \multirow{8}{*}{$\tau=0.25$}
        & HBQR-BL100 &\textbf{0.3675} &\textbf{0.2548} &0.9238& 0.8711 \\
        & HBQR-EN100 & \textbf{0.3795}  &\textbf{0.2585} & 0.9673 &  0.8725 \\
        & BQR-BL100 &0.3956 &0.2627 & 1.3406  & 0.9411 \\
        & BQR-EN100 &0.3859 & 0.2628  &0.9624 & 0.8821 \\

        & HBQR-BL200 &\textbf{0.3328} & \textbf{0.2108 }&0.6624 & 0.8483 \\
        & HBQR-EN200 &\textbf{0.3380} &\textbf{0.2104} &0.6822 & 0.8576 \\
        & BQR-BL200 &0.3534  & 0.2118 &0.9076 & 0.9311 \\
        & BQR-EN200 & 0.3476 & 0.2123 &0.6819 & 0.8732 \\

\hline
         \multirow{10}{*}{$\tau=0.5$} 
         & HBQR-BL100 &0.2659 &0.2059 &0.9465 & 0.9211 \\
        & HBQR-EN100 &0.2678 &0.2100  &0.9834   &  0.9289    \\
        & HBL100 &\textbf{0.2426}  &\textbf{0.1891} &0.9553 & 0.9432 \\ 
        & BQR-BL100 &\textbf{0.2468} &\textbf{0.1946} &1.2976 & 0.9848 \\
        & BQR-EN100 & 0.2502 &0.1968  &0.9838 &  0.9413 \\
    
        & HBQR-BL200 &0.1962 &0.1550 &0.6810 & 0.9143 \\
        & HBQR-EN200 &0.1952 &0.1549 &0.6927 & 0.9192 \\
        & HBL200 &\textbf{0.1777} &\textbf{0.1404} &0.6763 & 0.9384 \\ 
        & BQR-BL200 &\textbf{0.1825} &\textbf{ 0.1452  }&0.8756 & 0.9778 \\
        & BQR-EN200 &0.1841 &0.1460 &0.6900 & 0.9329 \\
\hline
         \multirow{8}{*}{$\tau=0.75$} 
         & HBQR-BL100 &\textbf{0.3635}  &\textbf{0.2521} &0.9395 &0.8756  \\
        & HBQR-EN100 &\textbf{0.3662} &\textbf{0.2503} &0.9891 & 0.8722 \\
        & BQR-BL100 &0.3943  & 0.2587 &1.3484 & 0.9440 \\
        & BQR-EN100 &0.3853  &0.2664 &0.9920 & 0.8859 \\
  
        & HBQR-BL200 &\textbf{0.3386} & \textbf{0.2141}  &0.6703 & 0.8573 \\
        & HBQR-EN200 & \textbf{0.3315} &\textbf{0.2105} &0.6895 & 0.8562 \\
        & BQR-BL200 &0.3571  & 0.2146  & 0.9053 & 0.9340 \\
        & BQR-EN200 &0.3512 & 0.2174 &0.6890 & 0.8659 \\
        \hline
    \end{tabular}
    \caption{Numerical results in Simulation 1.}
    \label{tab:model1}
\end{table}
\FloatBarrier

\begin{table}[H]
    \centering
    \begin{tabular}{|c|c|c|c|c|c|}
    \hline
         & Methods & RMSE & MMAD & AL & CP  \\
         \hline
        \multirow{8}{*}{$\tau=0.25$} 
        & HBQR-BL100 &\textbf{0.3822} &\textbf{ 0.2579} &1.1290 & 0.9046 \\
        & HBQR-EN100 &\textbf{0.4051} &\textbf{0.2767} &1.2135 & 0.9002 \\
        & BQR-BL100 &0.5643 &0.3385 &2.6992 & 0.9506 \\
        & BQR-EN100 &0.4950 &0.3100 &1.6428 &  0.9311 \\

        & HBQR-BL200 &\textbf{0.3418} &\textbf{0.2233} & 0.8011 & 0.8762 \\
        & HBQR-EN200 &\textbf{0.3528} &\textbf{0.2368} & 0.8484 & 0.8765 \\
        & BQR-BL200 &0.4588 &0.2540 &1.7800 & 0.9521 \\
        & BQR-EN200 &0.4221 &0.2436 &1.2130 &  0.9394 \\
\hline
         \multirow{10}{*}{$\tau=0.5$}
         & HBQR-BL100 &\textbf{0.2886} &\textbf{0.2203} &1.175 & 0.9533 \\
        & HBQR-EN100 & 0.2945 &0.2336 &1.2251 & 0.9522 \\
        & HBL100 &\textbf{0.2683} &\textbf{0.2013} &1.604 & 0.9946 \\ 
        & BQR-BL100 &0.2954 &0.2262 &2.333 & 0.9992 \\
        & BQR-EN100 &0.3273 & 0.2340 &1.5357 & 0.9733 \\

        & HBQR-BL200 &0.2060 &0.1591 &0.7990 & 0.9279 \\
        & HBQR-EN200 &0.2130 &0.1679 &0.8332 & 0.9297 \\
        & HBL200 &\textbf{0.1793}  & \textbf{0.1386}&1.0921 & 0.9943 \\ 
        & BQR-BL200 &\textbf{0.1926} &\textbf{0.1511} &1.4813 & 0.9992 \\
        & BQR-EN200 &0.1941 & 0.1522 &1.1176 & 0.9929 \\
\hline
         \multirow{8}{*}{$\tau=0.75$} 
         & HBQR-BL100 &\textbf{0.3791} &\textbf{0.2624} &1.1734 & 0.9066 \\
        & HBQR-EN100 &\textbf{0.3889} &\textbf{0.2822} &1.2615 & 0.9000 \\
        & BQR-BL100 &0.5761 &0.3468 &2.7564 & 0.9525 \\
        & BQR-EN100 &0.4697 &0.3086 &1.8026 & 0.9433 \\
           
        & HBQR-BL200 &\textbf{0.3324}  &\textbf{0.2188} &0.8013 & 0.8792 \\
        & HBQR-EN200 &\textbf{0.3352} &\textbf{0.2198} & 0.8442 & 0.8705 \\
        & BQR-BL200 &0.4530 & 0.2498  &1.7595 & 0.9521 \\
        & BQR-EN200 &0.4087 &0.2416 &1.2458 &  0.9437 \\
        \hline
        
    \end{tabular}
    \caption{Numerical results in Simulation 2.}
    \label{tab:model2}
\end{table}
\FloatBarrier

\begin{table}[H]
    \centering
    \begin{tabular}{|c|c|c|c|c|c|}
    \hline
         & Methods & RMSE & MMAD & AL & CP  \\
         \hline
        \multirow{8}{*}{$\tau=0.25$} 
        & HBQR-BL100 &\textbf{0.5676} &\textbf{0.4038 }&2.3005 &  0.9259 \\
        & HBQR-EN100 &\textbf{0.5805} & \textbf{0.4285} &2.5853 & 0.9268 \\
        & BQR-BL100 &0.6854 &0.4754  &5.0697 & 0.9524 \\
        & BQR-EN100 &0.6147 &0.4293 &2.9063 & 0.9430 \\
   
        & HBQR-BL200 &\textbf{0.5173}  &\textbf{0.3732}  &1.8027 &  0.9094 \\
        & HBQR-EN200 &\textbf{0.5319} & \textbf{0.3825} &2.0427 & 0.9060 \\
        & BQR-BL200 &0.5836  & 0.4033  &3.7568 & 0.9519 \\
        & BQR-EN200 &0.5542 &0.3872 &2.4114 & 0.9422 \\
\hline
         \multirow{10}{*}{$\tau=0.5$} 
         & HBQR-BL100 &0.4949 &0.3576  &2.370 & 0.9719 \\
        & HBQR-EN100 &0.4958 &0.3856 &2.583 & 0.9700 \\
        & HBL100 &\textbf{0.4525} &\textbf{0.3248} &3.199 &  0.9980 \\ 
        & BQR-BL100 &0.5023 &0.3671  &4.4369 & 0.9992 \\
        & BQR-EN100 &\textbf{0.4910} & 0.3566 &2.7858 & 0.9871 \\

        & HBQR-BL200 &0.4317 &0.3197 &1.8692 &  0.9611 \\
        & HBQR-EN200 &0.4317 &0.3178 &2.0474 &  0.9624 \\
        & HBL200 &\textbf{0.3707}  &\textbf{0.2719}  &2.4534 &  0.9965 \\ 
        & BQR-BL200 &0.4053 & 0.3041 &3.3309 & 0.9992 \\
        & BQR-EN200 &\textbf{0.3975}  &\textbf{0.2995} &2.3221 &  0.9921 \\
\hline
         \multirow{8}{*}{$\tau=0.75$} 
         & HBQR-BL100 &\textbf{0.5563} & \textbf{0.3993}  & 2.3546  & 0.9300 \\
        & HBQR-EN100 &\textbf{0.5911} &\textbf{0.4240} & 2.7935& 0.9321 \\
        & BQR-BL100 &0.6650 &0.4591 &4.9652 & 0.9531 \\
        & BQR-EN100 &0.5984 & 0.4318 &3.3417 & 0.9482 \\

        & HBQR-BL200 & \textbf{0.5241}  &\textbf{0.3856}  & 1.8819 & 0.9114 \\
        & HBQR-EN200 & \textbf{0.5390} &\textbf{0.3926} & 2.1538 & 0.9033 \\
        & BQR-BL200 &0.5786 &0.4008 &3.8405 & 0.9522 \\
        & BQR-EN200 & 0.5455 &0.3940 & 2.6987 & 0.9479 \\
        \hline
    \end{tabular}
    \caption{Numerical results in Simulation 3.}
    \label{tab:model3}
\end{table}
\FloatBarrier

\begin{table}[H]
    \centering
    \begin{tabular}{|c|c|c|c|c|c|}
    \hline
         & Methods & RMSE & MMAD & AL & CP  \\
         \hline
        \multirow{8}{*}{$\tau=0.25$}
        & HBQR-BL100 &\textbf{0.2705} &\textbf{0.1957} &1.0598 & 0.9416 \\
        & HBQR-EN100 &\textbf{0.2803} &\textbf{0.2093} &1.1574 &0.9408  \\
        & BQR-BL100 &0.3335 & 0.2542 &2.5406 & 0.9990 \\
        & BQR-EN100 &0.3300 &0.2362 &1.5557 & 0.9769 \\

        & HBQR-BL200 &0.2265  &\textbf{0.1572}  &0.6799 & 0.9038 \\
        & HBQR-EN200 &0.2287 &\textbf{0.1584 }&0.7199 & 0.9004 \\
        & BQR-BL200 &\textbf{0.2168} &0.1679  &1.5867 & 0.9979 \\
        & BQR-EN200 &\textbf{0.2143 }& 0.1636 &1.0815 & 0.9721 \\
\hline
         \multirow{10}{*}{$\tau=0.5$}
         & HBQR-BL100 & \textbf{0.4965 }&\textbf{0.2950} & 1.2265 & 0.9193 \\
        & HBQR-EN100 &\textbf{0.5048} &\textbf{0.3049} &1.2383 & 0.9213 \\
        & HBL100 & 0.5434 & 0.3111& 1.6510 &  0.9450 \\ 
        & BQR-BL100 & 0.5609 & 0.3312 & 2.2801 & 0.9503 \\
        & BQR-EN100 & 0.5184 & 0.3134 & 1.6965 &  0.9430 \\
   
        & HBQR-BL200 &\textbf{ 0.4407} & \textbf{0.2385 }&0.8315 & 0.9029 \\
        & HBQR-EN200 &\textbf{0.4466} &\textbf{0.2314} &0.8457 & 0.9006 \\
        & HBL200 &0.4927  &0.2461 & 1.1366& 0.9446 \\ 
        & BQR-BL200 &0.5061  &0.2571  &1.4968 & 0.9506 \\
        & BQR-EN200 &0.4835 & 0.2515 &1.1630 & 0.9444 \\
\hline
         \multirow{8}{*}{$\tau=0.75$}
         & HBQR-BL100 &\textbf{0.8388} & \textbf{0.4236} &1.4484 & 0.9070 \\
        & HBQR-EN100 & \textbf{0.8417} & \textbf{0.4460} &1.5382 & 0.9005  \\
        & BQR-BL100 &1.1330 &0.5333 &2.7729 & 0.9492 \\
        & BQR-EN100 &1.0732 &0.5514 &2.1302 & 0.9281 \\
   
        & HBQR-BL200 &\textbf{0.7868} &\textbf{0.3716} &1.0323 & 0.8829 \\
        & HBQR-EN200 &\textbf{0.7940} &\textbf{0.3874} &1.0884 & 0.8676 \\
        & BQR-BL200 &1.0651  &0.4624 &1.9991 & 0.9462 \\
        & BQR-EN200 &1.0206 &0.4729 &1.5216 &  0.9149 \\
        \hline
    \end{tabular}
    \caption{Numerical results in Simulation 4.}
    \label{tab:model4}
\end{table}
\FloatBarrier

\begin{table}[H]
    \centering
    \begin{tabular}{|c|c|c|c|c|c|}
    \hline
         & Methods & RMSE & MMAD & AL & CP  \\
         \hline
        \multirow{8}{*}{$\tau=0.25$} 
        & HBQR-BL100 & \textbf{0.4465} & \textbf{0.3012}  & 1.4890 & 0.9169 \\
        & HBQR-EN100 & \textbf{0.4460} & \textbf{0.3193} & 1.5992 & 0.9189 \\
        & BQR-BL100 & 0.8163 &0.4688 & 4.3100 & 0.9522  \\
        & BQR-EN100 & 0.7197 &0.3866 & 1.2336 & 0.8401 \\

        & HBQR-BL200 &\textbf{0.3741}  & \textbf{0.2456} &1.0247 & 0.9108 \\
        & HBQR-EN200 &\textbf{0.3926} &\textbf{0.2647 }&1.1039 & 0.9035 \\
        & BQR-BL200 &0.7137  &0.3841 &3.1585  & 0.9521\\
        & BQR-EN200 &0.6376  &0.3508 & 1.5062 & 0.8983 \\
\hline
         \multirow{10}{*}{$\tau=0.5$} 
         & HBQR-BL100 & \textbf{0.3292} &\textbf{0.2257} & 1.4712 & 0.9668 \\
        & HBQR-EN100 &\textbf{ 0.2469} & \textbf{0.2002}  & 1.0604 & 0.9616 \\
        & HBL100 & 0.4151 &0.2771 & 2.3541 & 0.9909 \\ 
        & BQR-BL100 & 0.4577 &0.3172  & 3.6592 & 0.9963 \\
        & BQR-EN100 & 0.6727 & 0.3435 & 0.7418 &  0.8184 \\

        & HBQR-BL200 &\textbf{0.2320} &\textbf{0.1780} &0.9866 & 0.9594 \\
        & HBQR-EN200 &\textbf{0.2495} &\textbf{0.1861} &1.0505 & 0.9570 \\
        & HBL200 &0.2861 &0.1996 &1.7512 & 0.9951 \\ 
        & BQR-BL200 &0.3155  &0.2263  &2.6611 & 0.9990 \\
        & BQR-EN200 &0.4242 &0.2551 &1.2599 & 0.9279 \\
\hline
         \multirow{8}{*}{$\tau=0.75$} 
         & HBQR-BL100 & \textbf{0.4315} &\textbf{0.2998} & 1.5445 & 0.9321 \\
        & HBQR-EN100 & \textbf{0.4356} & \textbf{0.3011} & 1.7046 & 0.9303 \\
        & BQR-BL100 & 0.7818 &0.4538 & 4.2976 & 0.9649 \\
        & BQR-EN100 & 0.6703 & 0.3859 & 1.7851 & 0.8868 \\

        & HBQR-BL200 & \textbf{0.3732}  &\textbf{0.2442} &1.0471 & 0.9089 \\
        & HBQR-EN200 & \textbf{0.3913} & \textbf{0.2607} &1.1202 & 0.9008 \\
        & BQR-BL200 & 0.7062  &0.3789  &3.2171 & 0.9554 \\
        & BQR-EN200 &0.5920 & 0.3386 &1.8842 & 0.9287 \\
        \hline
    \end{tabular}
    \caption{Numerical results in Simulation 5.}
    \label{tab:model5}
\end{table}
\FloatBarrier

\begin{table}[H]
    \centering
    \begin{tabular}{|c|c|c|c|c|c|}
    \hline
         & Methods & RMSE & MMAD & AL & CP  \\
         \hline
        \multirow{8}{*}{$\tau=0.25$} 
        & HBQR-BL100 & \textbf{0.2669} & \textbf{0.2034} & 1.0971 & 0.9503 \\
        & HBQR-EN100 & \textbf{0.2734} & \textbf{0.2141} &1.1773 & 0.9546 \\
        & BQR-BL100 & 0.3397 & 0.2591 & 2.2468& 0.9960 \\
        & BQR-EN100 & 0.3285 & 0.2438 & 1.4439 & 0.9639 \\

        & HBQR-BL200 & \textbf{0.1931}  & \textbf{0.1454} &0.6920 & 0.9203 \\
        & HBQR-EN200 & \textbf{0.1978} & \textbf{0.1466} &0.7353 & 0.9235 \\
        & BQR-BL200 &0.2025 & 0.1595 &1.4260 & 0.9984 \\
        & BQR-EN200 &0.2002 &0.1557 &0.9916 & 0.9814 \\
\hline
         \multirow{10}{*}{$\tau=0.5$} 
         & HBQR-BL100 & \textbf{0.2550} & \textbf{0.1937} & 1.0816 &  0.9516 \\
        & HBQR-EN100 & \textbf{0.2634} & \textbf{0.2050} &1.1773 & 0.9546 \\
        & HBL100 & 0.4862 &0.2947  & 1.5152 & 0.9384 \\ 
        & BQR-BL100 &0.3228 & 0.2480 & 2.1960 & 0.9970 \\
        & BQR-EN100 &0.3200 & 0.2370 & 1.4295 & 0.9698 \\
          
        & HBQR-BL200 & \textbf{0.4094} & \textbf{0.2328} & 0.8321 &  0.8971 \\
        & HBQR-EN200 & \textbf{0.4147} & \textbf{0.2355} & 0.8701 & 0.8959 \\
        & HBL200 & 0.4450 &  0.2384 & 1.0562 & 0.9379 \\ 
        & BQR-BL200 & 0.4584  & 0.2502 &1.3950 & 0.9498 \\
        & BQR-EN200 & 0.4392 & 0.2452  & 1.0774 &  0.9363 \\
\hline
         \multirow{8}{*}{$\tau=0.75$} 
         & HBQR-BL100 & \textbf{0.8115} & \textbf{0.4248}  & 1.4426 & 0.8998 \\
        & HBQR-EN100 & \textbf{0.8229} & \textbf{0.4399} & 1.5531 & 0.8946 \\
        & BQR-BL100 & 1.0427 &0.5067 & 2.5333 & 0.94682 \\
        & BQR-EN100 & 0.9903  &0.5243 & 1.8811& 0.9092 \\
   
        & HBQR-BL200 & \textbf{0.7661} & \textbf{0.3735} & 1.0210 & 0.8689 \\
        & HBQR-EN200 & \textbf{0.7734} & \textbf{0.3836} &  1.0743 & 0.8630 \\
        & BQR-BL200 &0.9763 & 0.4381 &1.8219 & 0.9437 \\
        & BQR-EN200 & 0.9407 & 0.4492 & 1.3898 &  0.9019 \\
        \hline
    \end{tabular}
    \caption{Numerical results in Simulation 6.}
    \label{tab:model6}
\end{table}
\FloatBarrier

\section{Real Data Analysis}\label{sec:realdatanalysis}

The robustness and efficiency of the Bayesian Huberised regularised quantile regression models are demonstrated via the analysis of two benchmarking datasets: Crime data and Top Gear data.  They have large outliers. For a better interpretation of the parameters and to put the regressors on the common scale, we standardised all the numerical predictors and response variables to have mean 0 and variance 1. Like in simulation studies, we also consider all the five methods of which we generated 10,000 posterior samples after discarding discarding the first 5,000 posterior samples as a burn-in. Then we report posterior medians of regression coefficients and their $95\%$ credible intervals. For brevity, we drop the names of predictors of the datasets and keep the corresponding number to indicate each predictor.  For BQR-BL, BQR-EN, HBQR-BL and HBQR-EN, we set the quantile levels as $\tau\in\{0.1,0.5,0.9\}$ for the Crime and Top Gear datasets. 

Since datasets may contain outliers, we adopt the following four criteria as measures of predictive accuracy; mean squared prediction error (MSPE), mean absolute prediction error (MAPE), mean Huber prediction error (MHPE) for $\delta=1.345$ and median of squared prediction error (MedSPE) via 10-fold cross validation. They are defined by $\text{MSPE}=10^{-1}\sum^{10}_{j=1} (\by_j - \bX_j^T\hat{\bbeta}^{(-j)})^2$, $\text{MAPE}=10^{-1} \sum^{10}_{j=1} \abs{\by_j - \bX_j^T\hat{\bbeta}^{(-j)}}$, $\text{MHPE}=10^{-1} \sum^{10}_{j=1} L^{Huber}_{\delta}(\by_j - \bX_j^T\hat{\bbeta}^{(-j)})$ and $\text{MedSPE}=\text{median}_{1\leq j\leq 10} (\by_j - \bX_j^T\hat{\bbeta}^{(-j)})^2$, where $L^{Huber}_{\delta}(\cdot)$ is defined by (\ref{eq:huber}),  $\hat{\bbeta}^{(-j)}$ is the posterior median based on dataset except for $j$th validation set, and $\by_j$ and $\bX_j$ are the response variables and covariate matrix based on the $j$th validation set, respectively. 

\begin{table}[H]
    \centering
    \begin{tabular}{|c|c|c|c|c|c|}
         \hline
         & Methods & MSPE & MAPE & MedSPE & MHPE \\
         \hline
        \multirow{4}{*}{$\tau=0.1$} 
        & HBQR-BL & 0.0226 & 0.1223 & \textbf{0.0044} & 0.0113 \\
        & HBQR-EN & \textbf{0.0156} & \textbf{0.1034} & \textbf{0.0071} & \textbf{0.0078}  \\
        & BQR-BL & \textbf{0.0157} & \textbf{0.1054} & 0.0137 & \textbf{0.0078} \\
        & BQR-EN &  0.0169 & 0.1285 & 0.0150 & 0.0085  \\
\hline
        \multirow{5}{*}{$\tau=0.5$} 
        & HBQR-BL & \textbf{0.0123} & \textbf{0.0946} & \textbf{0.0067} & \textbf{0.0061}  \\
        & HBQR-EN & \textbf{0.0121} & \textbf{0.0938} & \textbf{0.0073} & \textbf{0.0061}  \\
        & HBL     & 0.0304 & 0.1534 & 0.0173 & 0.0152  \\
        & BQR-BL & 0.0192 & 0.1008 & 0.0081 & 0.0096 \\
        & BQR-EN & 0.0250 & 0.1474 & 0.0185 & 0.0125 \\
\hline
        \multirow{4}{*}{$\tau=0.9$} 
        & HBQR-BL & \textbf{0.0400} & \textbf{0.1439} & 0\textbf{.0077} & \textbf{0.0200} \\
        & HBQR-EN & \textbf{0.0278} & \textbf{0.1346} & \textbf{0.0079} & \textbf{0.0139} \\
        & BQR-BL & 0.0453 & 0.1582 & 0.0106 & 0.0226  \\
        & BQR-EN & 0.0401 & 0.1629 & 0.0146 & 0.0200 \\
\hline
        
    \end{tabular}
    \caption{Mean squared prediction error (MSPE), mean absolute prediction error (MAPE), mean Huber prediction error (MHPE) for $\delta = 1.345$ and median of squared prediction error (MedSPE) for Crime data, computed from 10-fold cross-validation.}
    \label{tab:crime}
\end{table}
\FloatBarrier

\begin{figure}[H]
    \centering
    \includegraphics[width=0.9\textwidth]{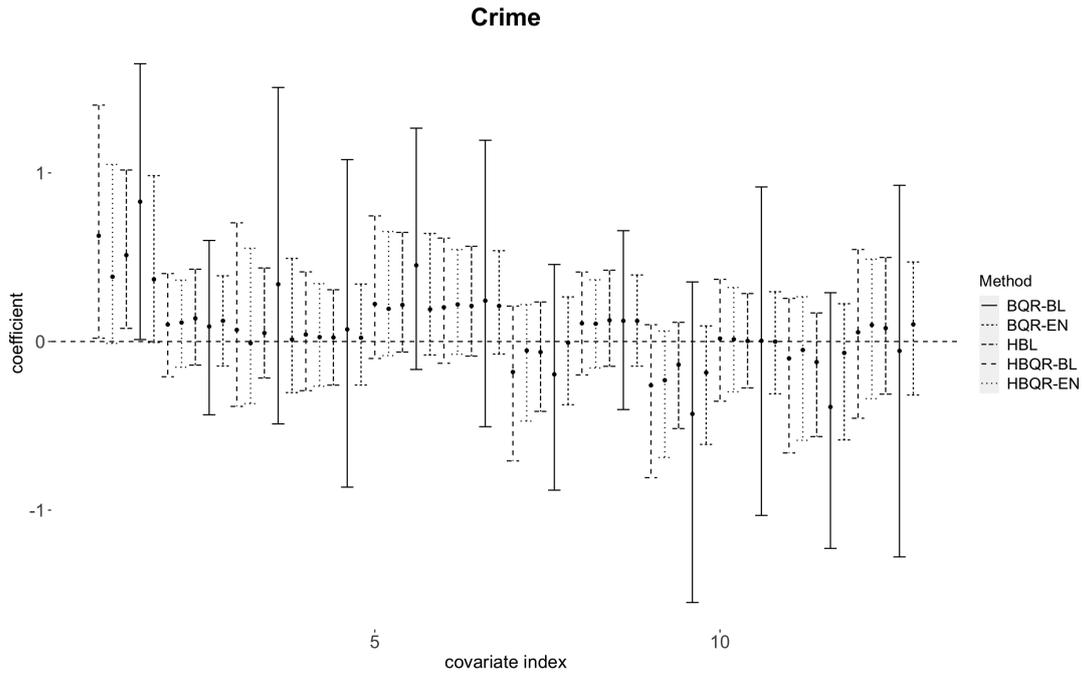}
    \caption{Posterior medians and 95\% credible intervals of the regression coefficients at $\tau=0.5$ in the Bayesian quantile regression with Bayesian lasso (BQR-BL), Bayesian quantile regression with elastic net (BQR-EN), the Huberized Bayesian lasso (HBL) and the proposed Bayesian quantile regression with Bayesian lasso (HBQR-BL) and elastic net (HBQR-EN), applied to the Crime data.}
    \label{fig:crimeCI5}
\end{figure}
\FloatBarrier

\begin{table}[H]
    \centering
    \begin{tabular}{|c|c|c|c|c|c|}
        \hline
         & Methods & MSPE & MAPE & MedSPE & MHPE \\
         \hline
        \multirow{4}{*}{$\tau=0.1$} 
        & HBQR-BL & \textbf{0.0288} & \textbf{0.1500} & \textbf{0.0196} & \textbf{0.0144}\\
        & HBQR-EN & \textbf{0.0296} & \textbf{0.1505} & \textbf{0.0229} & \textbf{0.0148}\\
        & BQR-BL & 0.0360 & 0.1736 &  0.0331 & 0.0180\\
        & BQR-EN &  0.0331 & 0.1605 &  0.0267 & 0.0166\\
\hline
        \multirow{5}{*}{$\tau=0.5$} 
        & HBQR-BL &0.0127  & 0.0942 & 0.0064  & 0.0064\\
        & HBQR-EN &  0.0120 & 0.0905 & 0.0070 &0.0060 \\
        & HBL     & \textbf{0.0110} &  \textbf{0.0863} & \textbf{0.0055} & \textbf{0.0055}\\
        & BQR-BL & 0.0183 & 0.1102 &  0.0101 & 0.0092\\
        & BQR-EN &  \textbf{0.0110}  & \textbf{0.0864} &  \textbf{0.0063} &  \textbf{0.0055}\\
\hline
        \multirow{4}{*}{$\tau=0.9$} 
        & HBQR-BL & \textbf{0.0643} & \textbf{0.2309} & \textbf{0.0410} & \textbf{0.0322} \\
        & HBQR-EN &  \textbf{0.0843} & \textbf{0.2662} & \textbf{0.0628} & \textbf{0.0421} \\
        & BQR-BL & 0.6942 & 0.7652 & 0.7337 & 0.3471 \\
        & BQR-EN &  0.2290 & 0.4461 & 0.1790 & 0.1145 \\
\hline
        
    \end{tabular}
    \caption{Mean squared prediction error (MSPE), mean absolute prediction error (MAPE), mean Huber prediction error (MHPE) for $\delta = 1.345$ and median of squared prediction error (MedSPE) for Top Gear data, computed from 10-fold cross-validation.}
    \label{tab:topgear}
\end{table}
\FloatBarrier

\begin{figure}[H]
    \centering
    \includegraphics[width=0.9\textwidth]{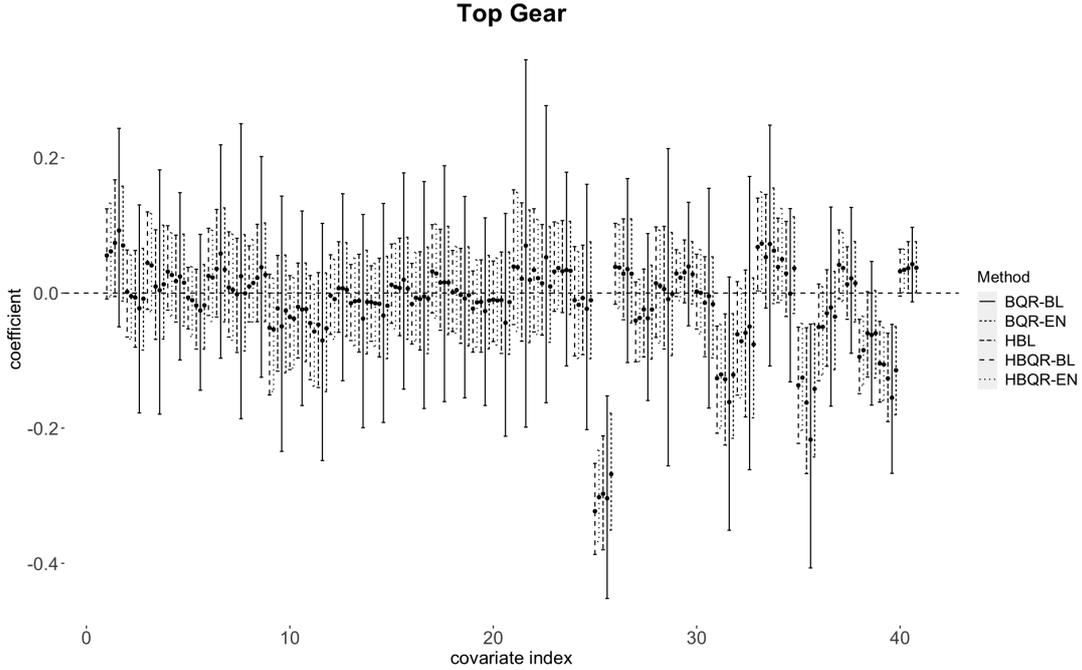}
    \caption{Posterior medians and 95\% credible intervals of the regression coefficients at $\tau=0.5$ in the Bayesian quantile regression with Bayesian lasso (BQR-BL), Bayesian quantile regression with elastic net (BQR-EN), the Huberized Bayesian lasso (HBL) and the proposed Bayesian quantile regression with Bayesian lasso (HBQR-BL) and elastic net (HBQR-EN), applied to the Top Gear data.}
    \label{fig:topgearCI5}
\end{figure}
\FloatBarrier

\subsection{Crime Dataset}

The data are collected from Statistical Abstract of the United States for the 50 states and the District of Columbia (\cite{statisticalUSA}). This data were analysed in the book of Statistical Methods for the Social Sciences (\cite{CrimeBook1997}). The predictors are the number of murders per 100,000 people in the population, the percentage of the population living in metropolitan areas, the percentage of the population who are white, the percentage of the population who are high school graduates or higher, the percentage of families living below the poverty level, and the percentage of families headed by a single parent (male householders with no wife present and with own children, or female householders with no husband present and with own children). The response of interest is the number of murders, forcible rapes, robberies, and aggravated assaults per 100,000 people in the population. In total, we have 51 observations and included squared variables, which results in 12 predictors in our models.

The posterior medians and $95\%$ credible intervals of the regression coefficients based on the five methods are reported in Figure \ref{fig:crimeCI5}. From the figure, all the methods  behave similarly and the estimation are very close.  The BQR-BL method produces relatively largest credible intervals, which suggests that this method may be unstable  in producing estimates. The similar performances can also be found in $\tau=0.1$ and $\tau=0.9$ (see Appendix C). Table \ref{tab:crime} also presents the predictive performance of the five methods for $\tau=0.5$ and the four Bayesian quantile regression based methods for $\tau\in\{0.1,0.9\}$. The proposed methods perform better than the existing robust methods in both median and upper quantile levels. The HBL method produces relatively large error measures in the median case among the rest of methods. Looking at the lower quantile level ($\tau=0.1$), MSPE, MAPE and MHPE suggest that HBQR-EN and BQR-BL perform better while MedSPE suggests that both proposed methods perform better. In this case, they are very comparable.

\subsection{Top Gear Dataset}

The data uses information on cars featuring on the website of the popular BBC television show Top Gear. It is available in the R package 'robustHD' (\cite{RobustHD}) and contains 242 observations on 29 numerical and categorical variables after removing the missing values. A description of the variables is provided in Table 3 of the paper (\cite{Alfons2016}). The response of interest is MPG (fuel consumption) and the remaining variables are predictors.  For categorical variables, there are 4 binary variables and 12 variables with three levels. These 12 variables are assigned two dummy variables each. The resulting design matrix consists of 12 numerical variables, 4 individual dummy variables, and 12 groups of two dummy variables each, giving a total of 40 predictors. 

The posterior medians and $95\%$ credible intervals of the regression coefficients based on the five methods are reported in Figure \ref{fig:topgearCI5}. From the figure, all the methods are comparable. Like for the Crime dataset, the BQR-BL method produces relatively largest credible intervals. The similar performances can also be found in $\tau=0.1$ and $\tau=0.9$ (see Appendix C).  Table \ref{tab:topgear} also presents the predictive performance of the five methods for $\tau=0.5$ and the four Bayesian quantile regression based methods for $\tau\in\{0.1,0.9\}$. All the methods are comparable in the median case ($\tau=0.5$) where both HBL and BQR-EN have the lowest error measures. Looking at the extreme quantile levels ($\tau=0.1,0.9$), the proposed methods significantly outperform the BQR-BL and BQR-EN methods especially at the upper quantile level where the existing robust methods perform slightly worse than the proposed methods. Furthermore, BQR-BL has the highest error measures in all cases.

\section{Conclusion}\label{sec:conclusion}

In this paper, we have presented the Bayesian Huberised regularisation. We proposed the Huberised asymmetric loss function and its corresponding probability density function that leads to a scale mixture of normal distribution with exponential and generalised inverse Gaussian mixing densities. This results in fully Bayesian hierarchical models for quantile regression and its Gibbs sampling algorithm with approximate Gibbs sampler for the data-dependent estimation of the robustness parameter. We have proved theoretically that the proposed Bayesian models yield a good posterior propriety and unimodality in their joint posterior density with conditional prior for the regression coefficients. Simulation studies and real data examples show that the proposed methods are effective in predictive accuracy and their robustness is evident under a wide range of scenarios. In many situations, the proposed methods outperform the existing Bayesian regularised quantile regression methods especially at the extreme quantile levels. Our proposed methods have proven to be robust in obtaining valuable results. 

\section{Acknowledgments}
This work is supported by the UK Engineering and Physical Sciences Research Council (EPSRC) grant 2295266 for the Brunel University London for Doctoral Training.

\appendix

\numberwithin{equation}{section}
\setcounter{equation}{0}

\section{Proofs}

\subsection{Proposition \ref{prop:tau}}\label{sec:appendixA}

We set $\mu=0$ and we wish to calculate $P(X\leq 0)$, that is,

\begin{align*}
    P(X\leq 0) &= \int^0_{-\infty} f_X(x) dx\\
    &= \frac{\eta\tau(1-\tau)e^{\eta}}{2\rho^2 (\eta+1)} \int^0_{-\infty} \exp\left\{ -\sqrt{\eta\left( \eta - \frac{x(1-\tau)}{\rho^2}\right)}\right\} dx\\
    &= \frac{\eta\tau(1-\tau)e^{\eta}}{2\rho^2 (\eta+1)} \int_0^{\infty} \exp\left\{ -\sqrt{\eta\left( \eta + \frac{x(1-\tau)}{\rho^2}\right)}\right\} dx\,.
\end{align*}
By letting $u=\sqrt{\eta\left( \eta + \frac{x(1-\tau)}{\rho^2}\right)}$, we have 
\begin{align*}
   P(X\leq 0) &= \frac{\eta\tau(1-\tau)e^{\eta}}{2\rho^2 (\eta+1)} \int_\eta^{\infty} e^{-u} \times \frac{2u\rho^2}{\eta(1-\tau)}du\\
   &=  \frac{\tau e^{\eta}}{ (\eta+1)} \int_\eta^{\infty} u e^{-u} du\\
   &= \frac{\tau e^{\eta}}{ (\eta+1)}  \left( \left[-ue^{-u}\right]^{\infty}_\eta + \int^\infty_\eta e^{-u} du \right)\\
   &= \frac{\tau e^{\eta}}{ (\eta+1)}  \left( \eta e^{-\eta} + \left[-e^{-u}\right]^{\infty}_\eta \right)\\
   &= \frac{\tau e^{\eta}}{ (\eta+1)}  \left( e^{-\eta}(\eta +1)  \right)\\
   &= \tau\,.
\end{align*}
On the other hand, it follows that $P(X>0)=1-\tau$. This completes the proof.

\subsection{Theorem \ref{prop:scalemixture}}\label{sec:appendixA2}

Let a, b be some real constants. By using the equality 
\begin{align}\label{eq:equality}
     \exp(-|ab|) = \int^\infty_0 \frac{a}{\sqrt{2\pi \sigma}} \exp\left\{ -\frac{1}{2} (a^2\sigma+b^2\sigma^{-1}) \right\} d\sigma\,,
\end{align}
(\cite{AndrewsEtAl1974}) and let $a=\sqrt{\frac{\eta}{\rhosq}}$ and $b=\sqrt{\eta+ \frac{\epsilon_i}{\rhosq} \left(\tau-I(\epsilon_i<0)\right)}$,
$f(\epsilon_i)$ can be expressed as a scale mixture of asymmetric Laplace (AL) and generalised inverse Gaussian (GIG) densities:
\begin{align*}
    &\frac{\eta\tau(1-\tau)e^{\eta}}{2\rhosq(\eta+1)} \exp\left\{ -  \sqrt{\eta\left(\eta+ \frac{\epsilon_i}{\rhosq} \left(\tau-I(\epsilon_i<0)\right)\right)}  \right\} \\
    &\quad\quad\propto \int^\infty_0 \mathcal{ALD} \left( \epsilon_i ; 0, 2\sigma_i,\tau  \right) GIG\left( \sigma_i ; 1, \sqrt{\frac{\eta}{\rhosq}}, \sqrt{\eta\rhosq}  \right) d\sigma_i\,,
\end{align*}
where $GIG(v|1,c,d)$ denotes the GIG distribution and its density is given by (7) and ALD.

The ALD can be expressed as a scale mixture of normal and exponential densities using the equality (Equation (\ref{eq:equality})) by letting $a=\frac{1}{\sqrt{4\sigma_i}}$, $b=\frac{\epsilon_i}{\sqrt{4\sigma_i}}$ and multiplying a factor of $\exp\left\{-\frac{(2\tau-1)\epsilon}{4\sigma_i} \right\}$ (\cite{KozumiKobayashi2011}). Therefore, $f(\epsilon_i)$ is  expressed as a normal scale mixture of exponential and generalised inverse Gaussian (GIG) densities:

\begin{align*}
    &\frac{\eta\tau(1-\tau)e^{\eta}}{2\rhosq(\eta+1)} \exp\left\{ -  \sqrt{\eta\left(\eta+ \frac{\epsilon_i}{\rhosq} \left(\tau-I(\epsilon_i<0)\right)\right)}  \right\} \\
    &\quad\quad\propto \int^\infty_0 \int^\infty_0 N(\epsilon_i|(1-2\tau)v_i,4\sigma_iv_i) Exp\left(v_i;\frac{\tau(1-\tau)}{2\sigma_i}\right) GIG\left( \sigma_i ; 1, \sqrt{\frac{\eta}{\rhosq}}, \sqrt{\eta\rhosq}  \right) d\sigma_i dv_i\,,
\end{align*}
where $N(\cdot)$ and $Exp(\cdot)$ are the normal and ezponential densities, respectively. 

\subsection{Proposition \ref{prop:posterior-properity_lasso}}\label{app:posterior-properity_lasso}

The overall posterior distribution is given by
\begin{align*}
    \pi &(\bbeta,\rhosq,\bv,\bsigma,\bs|\by ) \\
    &=\frac{\pi(\by|\bX,\bbeta,\bv,\bsigma )\pi(\bbeta|\bs,\rhosq)\pi(\bv|\bsigma)\pi(\sigma|\rhosq)\pi(\rhosq)\pi(\bs) }
    {\iiiiint \pi(\by|\bX,\bbeta,\bv,\bsigma )\pi(\bbeta|\bs,\rhosq)\pi(\bv|\bsigma)\pi(\sigma|\rhosq)\pi(\rhosq)\pi(\bs) d\bbeta d\bv d\bsigma d\bs d\rhosq}\,.    
\end{align*}

We show that the normalising constant of the posterior distribution is finite, that is, 
\begin{align*}
    \iiiiint \pi(\by|\bX,\bbeta,\bv,\bsigma )\pi(\bbeta|\bs,\rhosq)\pi(\bv|\bsigma)\pi(\sigma|\rhosq)\pi(\rhosq)\pi(\bs) d\bbeta d\bv d\bsigma d\bs d\rhosq <\infty\,.
\end{align*}

First, we consider the integral with respect to $\bbeta$. We have 
\begin{align*}
    \int \pi(\by&|\bX,\bbeta,\bv,\bsigma )\pi(\bbeta|\bs,\rhosq) d\bbeta  \\
    &= \int  \prod^n_{i=1} \frac{1}{\sqrt{8\pi\sigma_iv_i}} \exp\left\{ - \frac{(y_i-\bx_i\bbeta-(1-2\tau) v_i)^2}{8\sigma_i v_i} \right\}\\
    &\quad\quad\times \prod^k_{j=1} \frac{1}{\sqrt{2\pi\rhosq s_j}} \exp\left\{ -\frac{\beta^2_j}{2\rhosq s_j} \right\} d\bbeta \\
    &= \int (8\pi)^{-n/2}(2\pi)^{-k/2}{(\rhosq)}^{-k/2} \left(\prod^n_{i=1}\sigma_i \right)^{-1/2} \left(\prod^n_{i=1}v_i \right)^{-1/2} \left(\prod^k_{j=1}s_j \right)^{-1/2} \\
    &\quad\quad\times \exp\left\{-\frac{1}{2} (\by-\bX\bbeta-(1-2\tau)\bv)^T\bV^{-1}(\by-\bX\bbeta-(1-2\tau)\bv) \right\}\\
    &\quad\quad\times \exp\left\{ -\frac{1}{2\rhosq} \bbeta^T\bLambda^{-1}\bbeta \right\} d\bbeta\,,
\end{align*}
where $\bV=\diag(4\sigma_1 v_1,\ldots,4\sigma_n v_n)$ and $\bLambda=\diag(s_1,\ldots,s_k)$. In particular, we have 
\begin{align}
    \int \exp&\left\{-\frac{1}{2} (\by-\bX\bbeta-(1-2\tau)\bv)^T\bV^{-1}(\by-\bX\bbeta-(1-2\tau)\bv) \right\} \times \exp\left\{ -\frac{1}{2\rhosq} \bbeta^T\bLambda^{-1}\bbeta \right\} d\bbeta \nonumber \\
    &= \exp\left\{ -\frac{1}{2} (\by-(1-2\tau)\bv)^T\bV^{-1}(\by-(1-2\tau)\bv) \right\} \nonumber \\
    &\quad\quad\times \int \exp\left\{ -\frac{1}{2}\left(\bbeta^T\left(\bX^T\bV^{-1}\bX+\frac{1}{\rhosq}\bLambda^{-1} \right) \bbeta - 2\bbeta^T\bX^T\bV^{-1}(\by-(1-2\tau)\bv) \right) \right\}  d\bbeta \nonumber\\
    &=\exp\left\{ -\frac{1}{2} (\by-(1-2\tau)\bv)^T\bV^{-1}(\by-(1-2\tau)\bv) \right\} \nonumber \\
    &\quad\quad\times(2\pi)^{k/2} \left| \left(\bX^T\bV^{-1}\bX+\frac{1}{\rhosq}\bLambda^{-1} \right)^{-1} \right|^{1/2} \nonumber\\
    &=\exp\left\{ -\frac{1}{2} (\by-(1-2\tau)\bv)^T\bV^{-1}(\by-(1-2\tau)\bv) \right\} \nonumber \\
    &\quad\quad\times (2\pi)^{k/2} \left| \frac{1}{\rhosq}\bLambda^{-1}  \right|^{-1/2} \left| \bV \right|^{1/2} \left| \bV + \rhosq\bX \bLambda\bX^T \right|^{-1/2} \label{eq:matrix}\\
    &=\exp\left\{ -\frac{1}{2} (\by-(1-2\tau)\bv)^T\bV^{-1}(\by-(1-2\tau)\bv) \right\} \nonumber \\
    &\quad\quad\times (2\pi)^{k/2} 2^n {(\rhosq)}^{k/2} \left(\prod^k_{j=1} s_j \right)^{1/2} \left( \prod^n_{i=1} \sigma_i \right)^{1/2} \left( \prod^n_{i=1} v_i \right)^{1/2} \left| \bV + \rhosq\bX \bLambda\bX^T \right|^{-1/2}\,. \nonumber
\end{align}
The expression in (\ref{eq:matrix}) is due to the identity of $|I+AB|=|I+BA|$ (\cite{Henderson1981}).

Hence, we have 
\begin{align*}
    \int \pi(\by&|\bX,\bbeta,\bv,\bsigma )\pi(\bbeta|\bs,\rhosq) d\bbeta  \\
    &= (2\pi)^{-n/2}   \exp\left\{ -\frac{1}{2} (\by-(1-2\tau)\bv)^T\bV^{-1}(\by-(1-2\tau)\bv) \right\} \\
    &\quad\quad\times \left| \bV + \rhosq\bX \bLambda\bX^T \right|^{-1/2}\,.
\end{align*}

Next, we have 
\begin{align*}
    \iiiiint& \pi(\by|\bX,\bbeta,\bv,\bsigma )\pi(\bbeta|\bs,\rhosq)\pi(\bv|\bsigma)\pi(\sigma|\rhosq)\pi(\rhosq)\pi(\bs) d\bbeta d\bv d\bsigma d\bs d\rhosq \\
    &= \iiiint (2\pi)^{-n/2}   \exp\left\{ -\frac{1}{2} (\by-(1-2\tau)\bv)^T\bV^{-1}(\by-(1-2\tau)\bv) \right\} \\
    &\quad\quad\times \left| \bV + \rhosq\bX \bLambda\bX^T \right|^{-1/2}
     \prod^n_{i=1} \frac{\tau(1-\tau)}{2\sigma_i} \exp\left\{ -\frac{\tau(1-\tau)v_i}{2\sigma_i} \right\}
    \pi(\sigma|\rhosq)\pi(\rhosq)\pi(\bs)  d\bv d\bsigma d\bs d\rhosq  \\
    &\leq  \iiiint (2\pi)^{-n/2}   \exp\left\{ -\frac{1}{2} (\by-(1-2\tau)\bv)^T\bV^{-1}(\by-(1-2\tau)\bv) \right\} \left| \bV \right|^{-1/2}\\
    &\quad\quad\times 
     \prod^n_{i=1} \frac{\tau(1-\tau)}{2\sigma_i} \exp\left\{ -\frac{\tau(1-\tau)v_i}{2\sigma_i} \right\}
    \pi(\sigma|\rhosq)\pi(\rhosq)\pi(\bs)  d\bv d\bsigma d\bs d\rhosq\,,
\end{align*}
by using the fact that $|A+B|\geq |A|$ implies $|A+B|^{-1/2}\leq |A|^{-1/2}$ for a positive definite matrix $A$ and a semi-positive definite matrix $B$. 

Next, we consider the integral with respect to $\bv$. First, we have 
\begin{align*}
    \int &\left| \bV \right|^{-1/2}  \exp\left\{ -\frac{1}{2} (\by-(1-2\tau)\bv)^T\bV^{-1}(\by-(1-2\tau)\bv) \right\} \\
    &\quad\quad\times   \prod^n_{i=1} \frac{\tau(1-\tau)}{2\sigma_i} \exp\left\{ -\frac{\tau(1-\tau)v_i}{2\sigma_i} \right\} d\bv \\
    &= \int \left(\frac{\tau(1-\tau)}{2}\right)^n 2^{-n}  \left( \prod^n_{i=1} \sigma_i \right)^{-3/2} \left( \prod^n_{i=1} v_i \right)^{-1/2} \\
    &\quad\quad \times \prod^n_{i=1} \exp\left\{ - \frac{(y_i-(1-2\tau) v_i)^2}{8\sigma_i v_i} -\frac{\tau(1-\tau)v_i}{2\sigma_i} \right\} d\bv\\
    &=  \left(\frac{\tau(1-\tau)}{4}\right)^n  \left( \prod^n_{i=1} \sigma_i \right)^{-3/2}  \\
    &\quad\quad \times \int \prod^n_{i=1} v^{-1/2}_i \exp\left\{ - \frac{(y_i-(1-2\tau) v_i)^2}{8\sigma_i v_i} -\frac{\left(1-(1-2\tau)^2\right)v_i}{8\sigma_i} \right\} d\bv\\
    &= \left(\frac{\tau(1-\tau)}{4}\right)^n  \left( \prod^n_{i=1} \sigma_i \right)^{-3/2} \\
    &\quad\quad \times \int \prod^n_{i=1} v^{-1/2}_i \exp\left\{ -\frac{y^2_i}{8\sigma_iv_i} -\frac{(1-2\tau)y_i}{4\sigma_i}  -\frac{v_i}{8\sigma_i} \right\} d\bv\\
    &= \left(\frac{\tau(1-\tau)}{4}\right)^n  \left( \prod^n_{i=1} \sigma_i \right)^{-3/2} \\
    &\quad\quad \times \int \prod^n_{i=1} v^{-1/2}_i \exp\left\{ -\frac{1}{2} \left( \frac{v_i}{4\sigma_i} + \frac{y^2_i}{4\sigma_i v_i} \right)  \right\} \exp\left\{ -\frac{(1-2\tau)y_i}{4\sigma_i} \right\} d\bv\,.
\end{align*}

Letting $a^2=\frac{1}{4\sigma_i}$ and $b^2=\frac{y^2_i}{4\sigma_i}$ and using the equality (Equatiom (\ref{eq:equality})), we have 
\begin{align*}
     \int &\left| \bV \right|^{-1/2}  \exp\left\{ -\frac{1}{2} (\by-(1-2\tau)\bv)^T\bV^{-1}(\by-(1-2\tau)\bv) \right\} \\
    &\quad\quad\times   \prod^n_{i=1} \frac{\tau(1-\tau)}{2\sigma_i} \exp\left\{ -\frac{\tau(1-\tau)v_i}{2\sigma_i} \right\} d\bv \\
    &= \left(\frac{\tau(1-\tau)}{4}\right)^n  \left( \prod^n_{i=1} \sigma_i \right)^{-3/2} \exp\left\{ -\frac{(1-2\tau)y_i}{4\sigma_i} \right\} \\
    &\quad\quad \times \int \prod^n_{i=1} v^{-1/2}_i \exp\left\{ -\frac{1}{2} \left( a^2v_i + b^2v^{-1}_i \right)  \right\} d\bv \\
    &= \left(\frac{\tau(1-\tau)}{4}\right)^n  \left( \prod^n_{i=1} \sigma_i \right)^{-3/2} \exp\left\{ -\frac{(1-2\tau)y_i}{4\sigma_i} \right\} \\
    &\quad\quad \times \prod^n_{i=1} (2\pi)^{1/2}(4\sigma_i)^{1/2}  \exp\left\{ -\frac{|y_i|}{4\sigma_i} \right\}\\
    &= (2\pi)^{n/2} \left(\frac{\tau(1-\tau)}{2}\right)^n  \left( \prod^n_{i=1} \sigma_i \right)^{-1} \exp\left\{ -\frac{|y_i|+(1-2\tau)y_i}{4\sigma_i} \right\}\,.
\end{align*}

Hence, we have 
\begin{align*}
    \iiiiint& \pi(\by|\bX,\bbeta,\bv,\bsigma )\pi(\bbeta|\bs,\rhosq)\pi(\bv|\bsigma)\pi(\sigma|\rhosq)\pi(\rhosq)\pi(\bs) d\bbeta d\bv d\bsigma d\bs d\rhosq \\
    &\leq  \iiint \left(\frac{\tau(1-\tau)}{2}\right)^n  \left( \prod^n_{i=1} \sigma_i \right)^{-1} \exp\left\{ -\frac{|y_i|+(1-2\tau)y_i}{4\sigma_i} \right\}\\
    &\quad\quad\times \prod^n_{i=1} \frac{1}{2\rhosq K_1(\eta)} \exp\left\{ - \frac{\eta}{2} \left(\frac{\sigma_i}{\rhosq} + \frac{\rhosq}{\sigma_i} \right) \right\} \pi(\rhosq)\pi(\bs)  d\bsigma d\bs d\rhosq  \\
    &= \iiint \left(\frac{\tau(1-\tau)}{2}\right)^n \left( \frac{1}{2\rhosq K_1(\eta)}  \right)^n \\
    &\quad\quad\times \prod^n_{i=1} \sigma_i^{-1} \exp\left\{ - \frac{\eta}{2} \left(\frac{\sigma_i}{\rhosq} + \frac{\rhosq}{\sigma_i} \right) -\frac{|y_i|+(1-2\tau)y_i}{4\sigma_i} \right\} \pi(\rhosq)\pi(\bs)  d\bsigma d\bs d\rhosq\,. 
\end{align*}

Next, we consider the integral with respect to $\bsigma$. First, we have 
\begin{align*}
    \int & \prod^n_{i=1} \sigma_i^{-1} \exp\left\{ - \frac{\eta}{2} \left(\frac{\sigma_i}{\rhosq} + \frac{\rhosq}{\sigma_i} \right) -\frac{|y_i|+(1-2\tau)y_i}{4\sigma_i} \right\}  d\bsigma\\
    &= \int \prod^n_{i=1} \sigma_i^{-1} \exp\left\{ - \frac{1}{2} \left(\frac{\eta\sigma_i}{\rhosq} + \left( \eta\rhosq + \frac{|y_i|+(1-2\tau)y_i}{2} \right) \frac{1}{\sigma_i} \right)  \right\} )  d\bsigma\,.
\end{align*}

Letting $c^2=\frac{\eta}{\rhosq}$ and $d^2=\eta\rhosq + \frac{|y_i|+(1-2\tau)y_i}{2}$ and using the fact that
\begin{align*}
    K_\nu(cd) = \frac{1}{2} \left(\frac{c}{d} \right)^{-\nu} \int^\infty_0 x^{\nu-1} \exp\left\{ -\frac{1}{2}\left(c^2 x+\frac{d^2}{x} \right) \right\} dx\,,
\end{align*}
 
 we have 
\begin{align*}
    \int & \prod^n_{i=1} \sigma_i^{-1} \exp\left\{ - \frac{\eta}{2} \left(\frac{\sigma_i}{\rhosq} + \frac{\rhosq}{\sigma_i} \right) -\frac{|y_i|+(1-2\tau)y_i}{4\sigma_i} \right\}  d\bsigma\\
    &=\prod^n_{i=1} 2K_0\left( \sqrt{\frac{\eta}{\rhosq} \left(\eta\rhosq + \frac{|y_i|+(1-2\tau)y_i}{2} \right) } \right)\,,.
\end{align*}

Hence, we have 
\begin{align}
    \iiiiint& \pi(\by|\bX,\bbeta,\bv,\bsigma )\pi(\bbeta|\bs,\rhosq)\pi(\bv|\bsigma)\pi(\sigma|\rhosq)\pi(\rhosq)\pi(\bs) d\bbeta d\bv d\bsigma d\bs d\rhosq \nonumber \\
    &\leq  \iint  \left(\frac{\tau(1-\tau)}{2}\right)^n \left( \frac{1}{2\rhosq K_1(\eta)}  \right)^n 2^n \prod^n_{i=1} K_0\left( \sqrt{\frac{\eta}{\rhosq} \left(\eta\rhosq + \frac{|y_i|+(1-2\tau)y_i}{2} \right) } \right) \nonumber\\
    &\quad\quad\times \prod^k_{j=1}  \frac{\lambda^2_1}{2} \exp\left\{ -\frac{\lambda^2_1s_j}{2} \right\}  \pi(\rhosq) d\bs d\rhosq \nonumber\\
    &= \int \left(\frac{\tau(1-\tau)}{2}\right)^n \left( \frac{1}{\rhosq K_1(\eta)}  \right)^n \prod^n_{i=1} K_0\left( \sqrt{\frac{\eta}{\rhosq} \left(\eta\rhosq + \frac{|y_i|+(1-2\tau)y_i}{2} \right) } \right) \frac{1}{\rhosq} d\rhosq \nonumber\\
    &= \left(\frac{\tau(1-\tau)}{2K_1(\eta)}\right)^n \int {(\rhosq)}^{-(n+1)} \prod^n_{i=1} K_0\left( \sqrt{ \eta^2 + \eta\frac{|y_i|+(1-2\tau)y_i}{2\rhosq} } \right) d\rhosq\,.\label{eq:int_rhosq}
\end{align}

In Equation (\ref{eq:int_rhosq}), we note that the inequality
\begin{align*}
    \sqrt{ \eta^2 + \eta\frac{|y_i|+(1-2\tau)y_i}{2\rhosq} } \geq \sqrt{\frac{\eta}{\rhosq} \left( \frac{|y_i|+(1-2\tau)y_i}{2} \right) }\,,
\end{align*}
implies 
\begin{align*}
    K_0\left( \sqrt{ \eta^2 + \eta\frac{|y_i|+(1-2\tau)y_i}{2\rhosq} } \right) \leq K_0\left( \sqrt{\frac{\eta}{\rhosq} \left( \frac{|y_i|+(1-2\tau)y_i}{2} \right) } \right)\,,
\end{align*}
for any $\eta>0$ for $i=1,\ldots,n$. Hence, we have 

\begin{align*}
     \iiiiint& \pi(\by|\bX,\bbeta,\bv,\bsigma )\pi(\bbeta|\bs,\rhosq)\pi(\bv|\bsigma)\pi(\sigma|\rhosq)\pi(\rhosq)\pi(\bs) d\bbeta d\bv d\bsigma d\bs d\rhosq \nonumber \\
    &\leq  \left(\frac{\tau(1-\tau)}{2K_1(\eta)}\right)^n \int {(\rhosq)}^{-(n+1)} \prod^n_{i=1} K_0\left( \sqrt{ \eta\frac{|y_i|+(1-2\tau)y_i}{2\rhosq} } \right) d\rhosq\,.
\end{align*}

Using the fact that 
\begin{align*}
    K_0(x) < K_{1/2}(x) = \frac{\sqrt{\pi}e^{-x}}{\sqrt{2x}}\,,
\end{align*}
holds for all $x>0$ (\cite{YangChu2017}), we obtain 

\begin{align*}
    \iiiiint& \pi(\by|\bX,\bbeta,\bv,\bsigma )\pi(\bbeta|\bs,\rhosq)\pi(\bv|\bsigma)\pi(\sigma|\rhosq)\pi(\rhosq)\pi(\bs) d\bbeta d\bv d\bsigma d\bs d\rhosq \nonumber \\
    &<   \left(\frac{\tau(1-\tau)}{2K_1(\eta)}\right)^n \int {(\rhosq)}^{-(n+1)} \\
    &\quad\quad\quad\quad\quad\quad\quad\quad\times \prod^n_{i=1} \sqrt{\frac{\pi}{2}} \left( \eta\frac{|y_i|+(1-2\tau)y_i}{2\rhosq} \right)^{-1/4}  \exp\left\{ - \sqrt{\eta\frac{|y_i|+(1-2\tau)y_i}{2\rhosq}} \right\} d\rhosq\\
    &=  \left(\frac{\sqrt{\pi}\tau(1-\tau)}{2\sqrt{2}K_1(\eta)}\right)^n \left( \eta\frac{|y_i|+(1-2\tau)y_i}{2} \right)^{-1/4} \\
    &\quad\quad\times \int  {(\rhosq)}^{-(3n/2+1)}  \exp\left\{ - \frac{1}{\sqrt{\rhosq}}\sum^n_{i=1} \sqrt{\eta\frac{|y_i|+(1-2\tau)y_i}{2} } \right\} d\rhosq\,.
\end{align*}

By using the transformation $\sqrt{\rhosq}=x$, we have 
\begin{align}
     \int & {(\rhosq)}^{-(3n/2+1)}  \exp\left\{ - \frac{1}{\sqrt{\rhosq}}\sum^n_{i=1} \sqrt{\eta\frac{|y_i|+(1-2\tau)y_i}{2} } \right\} d\rhosq \nonumber\\
     &= 2 \int  x^{-3n/2-1}  \exp\left\{ - \frac{1}{x}\sqrt{\frac{\eta}{2}}\sum^n_{i=1} \sqrt{|y_i|+(1-2\tau)y_i } \right\} dx\,. \label{eq:integrand}
\end{align}
Since the integrand is the kernel of $IG\left(\frac{3n}{2}, \sqrt{\frac{\eta}{2}}\sum^n_{i=1} \sqrt{|y_i|+(1-2\tau)y_i } \right)$ where $IG(\cdot)$ is the inverse Gamma distribution, the integral is finite for any $n$. Hence, the posterior distribution under the improper prior $\pi(\rhosq)\propto \frac{1}{\rhosq}$ is proper for any $n$.

\subsection{Proposition \ref{prop:unimodality-lasso}}\label{app:unimodality-lasso}

The joint posterior density of $(\bbeta,\rhosq)$ is expressed by
\begin{align*}
    \pi(\bbeta,\rhosq|\by) = \iint \pi(\by|\bX,\bbeta,\bsigma,\bv)\pi(\bbeta|\rhosq)\pi(\bv|\bsigma)\pi(\bsigma|\rhosq)\pi(\rhosq) d\bv d\bsigma\,.
\end{align*}

First, we consider the integral with respect to $\bv$. We have 
\begin{align*}
    \int& \pi(\by|\bX,\bbeta,\bsigma,\bv)\pi(\bv|\bsigma) d\bv\\
    &=\int \prod^n_{i=1} \frac{1}{\sqrt{8\pi\sigma_iv_i}} \exp\left\{ - \frac{(y_i-\bx_i\bbeta-(1-2\tau) v_i)^2}{8\sigma_i v_i} \right\}\\
    &\quad\quad \times \prod^n_{i=1} \frac{\tau(1-\tau)}{2\sigma_i} \exp\left\{ -\frac{\tau(1-\tau)v_i}{2\sigma_i} \right\} d\bv \\
    &=(8\pi)^{-n/2} \left(\frac{\tau(1-\tau)}{2} \right)^n \left(\prod^n_{i=1}\sigma_i \right)^{-3/2} \\
    &\quad\quad\times \int \prod^n_{i=1} v^{-1/2}_i  \exp\left\{ - \frac{(y_i-\bx_i\bbeta-(1-2\tau) v_i)^2}{8\sigma_i v_i} -\frac{\tau(1-\tau)v_i}{2\sigma_i}  \right\} d\bv\\
    &=  \left(\frac{\tau(1-\tau)}{2} \right)^n \prod^n_{i=1}\sigma_i^{-1} \exp\left\{ - \frac{|y_i-\bx_i\bbeta|+(1-2\tau)(y_i-\bx_i\bbeta)}{4\sigma_i}\right\}\,.
\end{align*}

Hence, we have 
\begin{align*}
    \pi(\bbeta,\rhosq|\by) &= \pi(\bbeta|\rhosq)\pi(\rhosq) \int \left(\frac{\tau(1-\tau)}{2} \right)^n \prod^n_{i=1}\sigma_i^{-1} \exp\left\{ - \frac{|y_i-\bx_i\bbeta|+(1-2\tau)(y_i-\bx_i\bbeta)}{4\sigma_i}\right\} \\
    &\quad\quad\quad\times \prod^n_{i=1} \frac{1}{2\rhosq K_1(\eta)} \exp\left\{ - \frac{\eta}{2} \left(\frac{\sigma_i}{\rhosq} + \frac{\rhosq}{\sigma_i} \right) \right\} d\bsigma\\
    &= \pi(\bbeta|\rhosq)\pi(\rhosq) {(\rhosq)}^{-n}   \left(\frac{\tau(1-\tau)}{4K_1(\eta)} \right)^n \\
    &\quad\quad\times \int \prod^n_{i=1} \sigma_i^{-1} \exp\left\{ -\frac{1}{2}\left(\frac{\eta}{\rhosq\sigma_i} + \left(\eta\rhosq + \frac{|y_i-\bx_i\bbeta|+(1-2\tau)(y_i-\bx_i\bbeta)}{2} \right)\frac{1}{\sigma_i} \right) \right\}d\bsigma \\
    &= \pi(\bbeta|\rhosq)\pi(\rhosq) {(\rhosq)}^{-n}   \left(\frac{\tau(1-\tau)}{4K_1(\eta)} \right)^n \\
    &\quad\quad\times \prod^n_{i=1} 2K_0\left(\sqrt{\frac{\eta}{\rhosq}\left( \frac{|y_i-\bx_i\bbeta|+(1-2\tau)(y_i-\bx_i\bbeta)}{2}\right) } \right)\\
    &\propto {(\rhosq)}^{-1} {(\rhosq)}^{-k/2} {(\rhosq)}^{-n} \prod^k_{j=1} \exp\left\{ -\frac{\lambda_1|\beta_j|}{\sqrt{\rhosq}}\right\} \\
    &\quad\quad\times  \prod^n_{i=1} K_0\left(\sqrt{\frac{\eta}{\rhosq}\left( \frac{|y_i-\bx_i\bbeta|+(1-2\tau)(y_i-\bx_i\bbeta)}{2}\right) } \right)\\
    &={(\rhosq)}^{-n-k/2-1} \exp\left\{ -\frac{\lambda_1}{\sqrt{\rhosq}} \sum^k_{j=1} |\beta_j|\right\} \\
    &\quad\quad\times \prod^n_{i=1} K_0\left(\sqrt{\frac{\eta}{\rhosq}\left( \frac{|y_i-\bx_i\bbeta|+(1-2\tau)(y_i-\bx_i\bbeta)}{2}\right) } \right)\,.
\end{align*}

Then the log posterior density is given by 
\begin{align}
    \log \pi(\bbeta,\rhosq|\by) &= -\left(n+\frac{k}{2}+1 \right)\log\rhosq -\frac{\lambda_1}{\sqrt{\rhosq}} \norm{\bbeta}_1  \nonumber\\
    &\quad\quad
    +\sum^n_{i=1} \log\left[ K_0\left(\sqrt{\frac{\eta}{\rhosq}\left( \frac{|y_i-\bx_i\bbeta|+(1-2\tau)(y_i-\bx_i\bbeta)}{2}\right) } \right) \right]\,.
    \label{eq:logpost}
\end{align}

Like \cite{Kawakami2023} and \cite{CaiSun2021}, we consider the coordinate transformation $\Phi \leftrightarrow \frac{\bbeta}{\sqrt{\rhosq}}$, $\xi \leftrightarrow \frac{1}{\sqrt{\rhosq}}$. In the transformation coordinate, Equation (\ref{eq:logpost}) is given by
\begin{align}
    &(2n+k-2)\log\xi - \lambda_1\norm{\Phi}_1 \nonumber \\
    &\quad\quad 
    +\sum^n_{i=1} \log\left[ K_0\left(\sqrt{\eta^2 + \frac{\eta\xi}{2}\left( |\xi y_i-\bx_i\Phi|+(1-2\tau)(\xi y_i-\bx_i\Phi) \right) } \right) \right]\,.
    \label{eq:logcoord}
\end{align}

The first two terms in Equation (\ref{eq:logcoord}) are concave. The last term is also concave, since the Theorem 2(b) of \cite{BariczEtAl2011} is equivalent to log-convexity of $K_\nu$ for every $\nu$. Therefore, the joint posterior $\pi(\bbeta,\rhosq|\by)$ is unimodal. This completes the proof. 

\subsection{Proposition \ref{prop:posterior-properity_elastic}}\label{app:posterior-properity_elastic}

Like the proof of Proposition \ref{prop:posterior-properity_lasso}, we follow in the similar manner. 
The overall posterior distribution is given by
\begin{align*}
    \pi &(\bbeta,\rhosq,\bv,\bsigma,\bt|\by ) \\
    &=\frac{\pi(\by|\bX,\bbeta,\bv,\bsigma )\pi(\bbeta|\bt,\rhosq)\pi(\bv|\bsigma)\pi(\sigma|\rhosq)\pi(\rhosq)\pi(\bt) }
    {\iiiiint \pi(\by|\bX,\bbeta,\bv,\bsigma )\pi(\bbeta|\bt,\rhosq)\pi(\bv|\bsigma)\pi(\sigma|\rhosq)\pi(\rhosq)\pi(\bt) d\bbeta d\bv d\bsigma d\bt d\rhosq}\,.  
\end{align*}

We show that the normalising constant of the posterior distribution is finite, that is, 
\begin{align*}
    \iiiiint \pi(\by|\bX,\bbeta,\bv,\bsigma )\pi(\bbeta|\bt,\rhosq)\pi(\bv|\bsigma)\pi(\sigma|\rhosq)\pi(\rhosq)\pi(\bt) d\bbeta d\bv d\bsigma d\bt d\rhosq <\infty\,.
\end{align*}

First, we consider the integral with respect to $\bbeta$. We have 
\begin{align*}
    \int &\pi(\by|\bX,\bbeta,\bv,\bsigma )\pi(\bbeta|\bt,\rhosq) d\bbeta\\
    &=\int (8\pi)^{-n/2}(\pi)^{-k/2}\lambda_4^{k/2} {(\rhosq)}^{-k/2} \left(\prod^n_{i=1}\sigma_i \right)^{-1/2} \left(\prod^n_{i=1}v_i \right)^{-1/2} \left(\prod^k_{j=1} \frac{t_j}{t_j-1} \right)^{1/2} \\
    &\quad\quad\times \exp\left\{-\frac{1}{2} (\by-\bX\bbeta-(1-2\tau)\bv)^T\bV^{-1}(\by-\bX\bbeta-(1-2\tau)\bv) \right\}\\
    &\quad\quad\times \exp\left\{ -\frac{\lambda_4}{\rhosq} \bbeta^T\bLambda^{-1}_2\bbeta \right\} d\bbeta\,,
\end{align*}
where $\bV=\diag(4\sigma_1 v_1,\ldots,4\sigma_n v_n)$ and $\bLambda_2=\diag((t_1-1)t_1^{-1},\ldots,(t_n-1)t_n^{-1})$. In particular, we have 
\begin{align*}
    \int \exp&\left\{-\frac{1}{2} (\by-\bX\bbeta-(1-2\tau)\bv)^T\bV^{-1}(\by-\bX\bbeta-(1-2\tau)\bv) \right\} \times \exp\left\{ -\frac{\lambda_4}{\rhosq} \bbeta^T\bLambda^{-1}_2\bbeta \right\} d\bbeta\\
    &=\exp\left\{ -\frac{1}{2} (\by-(1-2\tau)\bv)^T\bV^{-1}(\by-(1-2\tau)\bv) \right\} \\
    &\quad\quad\times(2\pi)^{k/2} \left| \left(\bX^T\bV^{-1}\bX+ \frac{2\lambda_4}{\rhosq} \bLambda^{-1}_2 \right)^{-1} \right|^{1/2} \\
    &=\exp\left\{ -\frac{1}{2} (\by-(1-2\tau)\bv)^T\bV^{-1}(\by-(1-2\tau)\bv) \right\} \\
    &\quad\quad\times (2\pi)^{k/2} \left|\frac{2\lambda_4}{\rhosq} \bLambda^{-1}_2  \right|^{-1/2} \left| \bV \right|^{1/2} \left| \bV + \frac{\rhosq}{2\lambda_4} \bX\bLambda^{-1}_2 \bX^T \right|^{-1/2}\\
    &=\exp\left\{ -\frac{1}{2} (\by-(1-2\tau)\bv)^T\bV^{-1}(\by-(1-2\tau)\bv) \right\} \\
    &\quad\quad\times (2\pi)^{k/2} 2^n 2^{k/2} {(\rhosq)}^{k/2} \lambda_4^{-k/2} \left(\prod^k_{j=1} \frac{t_j}{t_j-1} \right)^{-1/2} \left( \prod^n_{i=1} \sigma_i \right)^{1/2} \left( \prod^n_{i=1} v_i \right)^{1/2} \\
    &\quad\quad\times \left| \bV + \frac{\rhosq}{2\lambda_4} \bX\bLambda^{-1}_2 \bX^T \right|^{-1/2}\,.
\end{align*}

Hence, we have 
\begin{align*}
     \int &\pi(\by|\bX,\bbeta,\bv,\bsigma )\pi(\bbeta|\bt,\rhosq) d\bbeta\\
    &= (2\pi)^{n/2}  \exp\left\{-\frac{1}{2} (\by-(1-2\tau)\bv)^T\bV^{-1}(\by-(1-2\tau)\bv) \right\}\\
    &\quad\quad\times \left| \bV + \frac{\rhosq}{2\lambda_4} \bX\bLambda^{-1}_2 \bX^T \right|^{-1/2}\,.
\end{align*}

Next, we have 
\begin{align*}
     \iiiiint& \pi(\by|\bX,\bbeta,\bv,\bsigma )\pi(\bbeta|\bt,\rhosq)\pi(\bv|\bsigma)\pi(\sigma|\rhosq)\pi(\rhosq)\pi(\bt) d\bbeta d\bv d\bsigma d\bt d\rhosq\\
     &=\iiiint (2\pi)^{n/2}  \exp\left\{-\frac{1}{2} (\by-(1-2\tau)\bv)^T\bV^{-1}(\by-(1-2\tau)\bv) \right\}\\
    &\quad\quad\times \left| \bV + \frac{\rhosq}{2\lambda_4} \bX\bLambda^{-1}_2 \bX^T \right|^{-1/2} \pi(\bv|\bsigma)\pi(\sigma|\rhosq)\pi(\rhosq)\pi(\bt) d\bv d\bsigma d\bt d\rhosq\\
    &\leq  \iiiint (2\pi)^{-n/2}   \exp\left\{ -\frac{1}{2} (\by-(1-2\tau)\bv)^T\bV^{-1}(\by-(1-2\tau)\bv) \right\} \left| \bV \right|^{-1/2}\\
    &\quad\quad\times 
     \prod^n_{i=1} \frac{\tau(1-\tau)}{2\sigma_i} \exp\left\{ -\frac{\tau(1-\tau)v_i}{2\sigma_i} \right\}
    \pi(\sigma|\rhosq)\pi(\rhosq)\pi(\bs)  d\bv d\bsigma d\bs d\rhosq\,.
\end{align*}

Next, we consider the integral with respect to $\bv$. We have
\begin{align*}
     \int &\left| \bV \right|^{-1/2}  \exp\left\{ -\frac{1}{2} (\by-(1-2\tau)\bv)^T\bV^{-1}(\by-(1-2\tau)\bv) \right\} \\
    &\quad\quad\times   \prod^n_{i=1} \frac{\tau(1-\tau)}{2\sigma_i} \exp\left\{ -\frac{\tau(1-\tau)v_i}{2\sigma_i} \right\} d\bv \\
    &= (2\pi)^{n/2} \left(\frac{\tau(1-\tau)}{2}\right)^n  \left( \prod^n_{i=1} \sigma_i \right)^{-1} \exp\left\{ -\frac{|y_i|+(1-2\tau)y_i}{4\sigma_i} \right\}\,.
\end{align*}

Hence, we have 
\begin{align*}
    \iiiiint& \pi(\by|\bX,\bbeta,\bv,\bsigma )\pi(\bbeta|\bt,\rhosq)\pi(\bv|\bsigma)\pi(\sigma|\rhosq)\pi(\rhosq)\pi(\bt) d\bbeta d\bv d\bsigma d\bt d\rhosq\\
    &\leq  \iiint \left(\frac{\tau(1-\tau)}{2}\right)^n  \left( \prod^n_{i=1} \sigma_i \right)^{-1} \exp\left\{ -\frac{|y_i|+(1-2\tau)y_i}{4\sigma_i} \right\}
    \\
    &\quad\quad\times  \prod^n_{i=1} \frac{1}{2\rhosq K_1(\eta)} \exp\left\{ - \frac{\eta}{2} \left(\frac{\sigma_i}{\rhosq} + \frac{\rhosq}{\sigma_i} \right) \right\}  \pi(\rhosq)\pi(\bt)   d\bsigma d\bt d\rhosq\\
    &=\iiint \left(\frac{\tau(1-\tau)}{4K_1(\eta)}\right)^n {(\rhosq)}^{-n}  \\
    &\quad\quad\times \prod^n_{i=1} \sigma_i  \exp\left\{ - \frac{1}{2} \left(\frac{\eta\sigma_i}{\rhosq} + \left(\eta\rhosq + \frac{|y_i|+(1-2\tau)y_i}{2}\right)\frac{1}{\sigma_i} \right) \right\}  \pi(\rhosq)\pi(\bt)   d\bsigma d\bt d\rhosq\\
    &=\iint \left(\frac{\tau(1-\tau)}{2K_1(\eta)}\right)^n {(\rhosq)}^{-n}  \prod^n_{i=1} K_0\left(\sqrt{ \frac{\eta}{\rhosq}\left(\eta\rhosq + \frac{|y_i|+(1-2\tau)y_i}{2} \right) } \right) \\
    &\quad\quad\times \prod^k_{j=1} \Gamma^{-1} \left(\frac{1}{2},\tlambda \right) \sqrt{\frac{\tlambda}{t_j}} \exp\left\{-\tlambda t_j \right\} I(t_j>1) \times \frac{1}{\rhosq} d\bt d\rhosq\\
    &= \left(\frac{\tau(1-\tau)}{2K_1(\eta)}\right)^n \tlambda^{-k}  \Gamma^{-k} \left(\frac{1}{2},\tlambda \right) \Gamma^k \left(\frac{1}{2},1 \right)\\
    &\quad\quad\times \int {(\rhosq)}^{-n-1} 
    \prod^n_{i=1} K_0\left(\sqrt{ \frac{\eta}{\rhosq}\left(\eta\rhosq + \frac{|y_i|+(1-2\tau)y_i}{2} \right) } \right) d\rhosq \\
    &<\left(\frac{\tau(1-\tau)}{2K_1(\eta)}\right)^n \tlambda^{-k}  \Gamma^{-k} \left(\frac{1}{2},\tlambda \right) \Gamma^k \left(\frac{1}{2},1 \right)\\
    &\quad\quad\times 2 \int x^{-3n/2-1} \exp\left\{ - \frac{1}{x}\sqrt{\frac{\eta}{2}}\sum^n_{i=1} \sqrt{|y_i|+(1-2\tau)y_i } \right\} dx\,.
\end{align*}

As the integrand is same as that in (\ref{eq:integrand}), the integral is finite for any $n$. Hence, the posterior distribution under the improper prior $\pi(\rhosq)\propto \frac{1}{\rhosq}$ is proper for any $n$.

\subsection{Proposition \ref{prop:unimodality-elastic}}\label{app:unimodality-elastic}

Like the proof of Proposition \ref{prop:unimodality-lasso}, we follow in the similar manner. The joint posterior density of $(\bbeta,\rhosq)$ is expressed by
\begin{align*}
    \pi(\bbeta,\rhosq|\by) &= \iint \pi(\by|\bX,\bbeta,\bsigma,\bv)\pi(\bbeta|\rhosq)\pi(\bv|\bsigma)\pi(\bsigma|\rhosq)\pi(\rhosq) d\bv d\bsigma\\
    &=\pi(\bsigma|\rhosq)\pi(\rhosq) \left(\frac{\tau(1-\tau)}{4K_1(\eta)} \right)^n {(\rhosq)}^{-n} \\
    &\quad\quad\times \prod^n_{i=1} K_0\left(\sqrt{\frac{\eta}{\rhosq}\left( \frac{|y_i-\bx_i\bbeta|+(1-2\tau)(y_i-\bx_i\bbeta)}{2}\right) } \right)\\
    &\propto{(\rhosq)}^{-n-k/2-1} \exp\left\{ -\frac{\lambda_3}{\sqrt{\rhosq}} \sum^k_{j=1} |\beta_j| -\frac{\lambda_4}{\rhosq}\sum^k_{j=1} \beta_j^2 \right\} \\
    &\quad\quad\times \prod^n_{i=1} K_0\left(\sqrt{\frac{\eta}{\rhosq}\left( \frac{|y_i-\bx_i\bbeta|+(1-2\tau)(y_i-\bx_i\bbeta)}{2}\right) } \right)\,.
\end{align*}

Then the log posterior density is given by 
\begin{align}
    \log \pi(\bbeta,\rhosq|\by) &=  -\left(n+\frac{k}{2}+1 \right)\log\rhosq -\frac{\lambda_3}{\sqrt{\rhosq}} \norm{\bbeta}_1 -\frac{\lambda_4}{\rhosq} \norm{\bbeta}^2_2 \nonumber\\
    &\quad\quad
    +\sum^n_{i=1} \log\left[ K_0\left(\sqrt{\frac{\eta}{\rhosq}\left( \frac{|y_i-\bx_i\bbeta|+(1-2\tau)(y_i-\bx_i\bbeta)}{2}\right) } \right) \right]\,.
    \label{eq:logpost-en}
\end{align}

We also consider the coordinate transformation $\Phi \leftrightarrow \frac{\bbeta}{\sqrt{\rhosq}}$, $\xi \leftrightarrow \frac{1}{\sqrt{\rhosq}}$. In the transformation coordinate, Equation (\ref{eq:logpost-en}) is given by
\begin{align*}
    &(2n+k+2)\log\xi - \lambda_1\norm{\Phi}_1 -\lambda_4\norm{\Phi}^2_2 \nonumber \\
    &\quad\quad 
    +\sum^n_{i=1} \log\left[ K_0\left(\sqrt{\eta^2 + \frac{\eta\xi}{2}\left( |\xi y_i-\bx_i\Phi|+(1-2\tau)(\xi y_i-\bx_i\Phi) \right) } \right) \right]\,.
\end{align*}

Since the four terms are log-concave, the joint posterior of $\pi(\bbeta,\rhosq|\by)$ is unimodal. This completes the proof. 

\section{Details of Gibbs Sampling Algorithm}
\subsection{Bayesian Huberised Lasso Quantile Regression}\label{app:gibbs-lasso}

The joint posterior distribution is as follows.
\begin{align*}
    \pi (\bbeta,\rhosq,&\bv,\bsigma,\lambda_1,\bs|\by )\\
    & = \prod^n_{i=1} \frac{1}{\sqrt{8\pi\sigma_iv_i}} \exp\left\{ - \frac{(y_i-\bx_i\bbeta-(1-2\tau) v_i)^2}{8\sigma_i v_i} \right\}\\
    &\quad \quad \times \prod^n_{i=1} \frac{1}{2\rhosq K_1(\eta)} \exp\left\{ - \frac{\eta}{2} \left(\frac{\sigma_i}{\rhosq} + \frac{\rhosq}{\sigma_i} \right) \right\}\\
    &\quad\quad \times \prod^n_{i=1} \frac{\tau(1-\tau)}{2\sigma_i} \exp\left\{ -\frac{\tau(1-\tau)v_i}{2\sigma_i} \right\}\\
    &\quad\quad\times \prod^k_{j=1} \frac{1}{\sqrt{2\pi\rhosq s_j}} \exp\left\{ -\frac{\beta^2_j}{2\rhosq s_j} \right\} \\
    &\quad\quad\times \prod^k_{j=1}  \frac{\lambda^2_1}{2} \exp\left\{ -\frac{\lambda^2_1s_j}{2} \right\} \\
    &\quad\quad\times \frac{b^a}{\Gamma(a)} {(\lambda^2_1)}^{a-1} \exp\left\{-b\lambda^2_1 \right\}\\
    &\quad\quad\times \frac{d^c}{\Gamma(c)} \eta^{c-1} \exp\left\{-d\eta \right\} \\
    &\quad\quad\times \frac{1}{\rhosq}\,.
\end{align*}

The full conditional posterior distribution of $\bbeta$ is given by  
\begin{align*}
    \pi(\bbeta|\by,\rhosq,&\bv,\bsigma,\lambda_1,\bs)\\
    &\propto \prod^n_{i=1} \frac{1}{\sqrt{8\pi\sigma_iv_i}} \exp\left\{ - \frac{(y_i-\bx_i\bbeta-(1-2\tau) v_i)^2}{8\sigma_i v_i} \right\}\\
    &\quad\quad\times \prod^k_{j=1} \frac{1}{\sqrt{2\pi\rhosq s_j}} \exp\left\{ -\frac{\beta^2_j}{2\rhosq s_j} \right\} \\
    &\propto \exp\left\{-\frac{1}{2} (\by-\bX\bbeta-(1-2\tau)\bv)^T\bV^{-1}(\by-\bX\bbeta-(1-2\tau)\bv) \right\}\\
    &\quad\quad\times \exp\left\{ -\frac{1}{2\rhosq} \bbeta^T\bLambda^{-1}\bbeta \right\}\\
    &\propto \exp\left\{ -\frac{1}{2}\left(\bbeta^T\left(\bX^T\bV^{-1}\bX+\frac{1}{\rhosq}\bLambda^{-1} \right) \bbeta - 2\bbeta^T\bX^T\bV^{-1}(\by-(1-2\tau)\bv) \right) \right\}\\
    &\propto N\left(\bmu_{\bbeta},\bSigma_{\bbeta} \right)\,,
\end{align*}
where $\bV=\diag(4\sigma_1 v_1,\ldots,4\sigma_n v_n)$, $\bLambda=\diag(s_1,\ldots,s_k)$, $\bSigma_{\bbeta}=\left(\bX^T\bV^{-1}\bX+\frac{1}{\rhosq}\bLambda^{-1} \right)^{-1}$ and $\bmu_{\bbeta}=\bSigma_{\bbeta} \bX^T\bV^{-1}(\by-(1-2\tau)\bv)$. 

The full conditional posterior distribution of $\sigma_i$, $i=1,\ldots,n$, is given by
\begin{align*}
    \pi(\sigma_i|\by,\bbeta,\rhosq,&\bv,\lambda_1,\bs)\\
    &\propto \prod^n_{i=1} \frac{1}{\sqrt{8\pi\sigma_iv_i}} \exp\left\{ - \frac{(y_i-\bx_i\bbeta-(1-2\tau) v_i)^2}{8\sigma_i v_i} \right\}\\
    &\quad \quad \times \prod^n_{i=1} \frac{1}{2\rhosq K_1(\eta)} \exp\left\{ - \frac{\eta}{2} \left(\frac{\sigma_i}{\rhosq} + \frac{\rhosq}{\sigma_i} \right) \right\}\\
    &\quad\quad \times \prod^n_{i=1} \frac{\tau(1-\tau)}{2\sigma_i} \exp\left\{ -\frac{\tau(1-\tau)v_i}{2\sigma_i} \right\}\\
    &\propto \sigma_i^{-3/2} \exp\left\{-\frac{1}{2}\left( \frac{\eta}{\rhosq}\sigma_i + \left( \frac{(y_i-\bx_i\bbeta-(1-2\tau) v_i)^2}{4 v_i} + \tau(1-\tau)v_i +\eta\rhosq \right)\frac{1}{\sigma_i} \right) \right\}\\
    &\propto GIG\left(-\frac{1}{2}, \frac{\eta}{\rhosq}, \frac{(y_i-\bx_i\bbeta-(1-2\tau) v_i)^2}{4 v_i} + \tau(1-\tau)v_i +\eta\rhosq  \right)\,.
\end{align*}

The full conditional posterior distribution of $v_i$, $i=1,\ldots,n$, is given by
\begin{align*}
    \pi(v_i|\by,\bbeta,\rhosq,&\bsigma,\lambda_1,\bs)\\
    &\propto \prod^n_{i=1} \frac{1}{\sqrt{8\pi\sigma_iv_i}} \exp\left\{ - \frac{(y_i-\bx_i\bbeta-(1-2\tau) v_i)^2}{8\sigma_i v_i} \right\}\\
    &\quad\quad \times \prod^n_{i=1} \frac{\tau(1-\tau)}{2\sigma_i} \exp\left\{ -\frac{\tau(1-\tau)v_i}{2\sigma_i} \right\}\\
    &\propto v^{-1/2}_i exp\left\{-\frac{1}{2} \left( \frac{(y_i-\bx_i\bbeta-(1-2\tau) v_i)^2}{4\sigma_i v_i} + \frac{\tau(1-\tau)v_i}{\sigma_i} \right) \right\}\\
    &\propto v^{-1/2}_i exp\left\{-\frac{1}{2} \left( \frac{(y_i-\bx_i\bbeta )^2}{4\sigma_i}\frac{1}{v_i} + \left(\frac{(1-2\tau)^2}{4\sigma_i} + \frac{\tau(1-\tau)}{\sigma_i} \right)v_i \right) \right\}\\
    &\propto GIG\left(\frac{1}{2}, \frac{(1-2\tau)^2}{4\sigma_i} + \frac{\tau(1-\tau)}{\sigma_i}, \frac{(y_i-\bx_i\bbeta )^2}{4\sigma_i} \right)\,.
\end{align*}

The full conditional posterior distribution of $\rhosq$ is given by
\begin{align*}
    \pi(\rhosq|\by,\bbeta,\bv,&\bsigma,\lambda_1,\bs)\\
    &\propto \prod^n_{i=1} \frac{1}{2\rhosq K_1(\eta)} \exp\left\{ - \frac{\eta}{2} \left(\frac{\sigma_i}{\rhosq} + \frac{\rhosq}{\sigma_i} \right) \right\}\\
    &\quad\quad\times \prod^k_{j=1} \frac{1}{\sqrt{2\pi\rhosq s_j}} \exp\left\{ -\frac{\beta^2_j}{2\rhosq s_j} \right\} \\
    &\quad\quad\times \frac{1}{\rhosq}\\
    &\propto {(\rhosq)}^{-n-\frac{k}{2}-1} \exp\left\{ \frac{1}{2}\left( \sum^n_{i=1}\frac{\eta}{\sigma_i} \rhosq + \left(\sum^n_{i=1}\eta\sigma_i + \sum_{j=1}^k \frac{\beta^2_j}{s_j} \right)\frac{1}{\rhosq} \right) \right\}\\
    &\propto GIG\left( -n-\frac{k}{2}, \sum^n_{i=1}\frac{\eta}{\sigma_i}, \sum^n_{i=1}\eta\sigma_i + \sum_{j=1}^k \frac{\beta^2_j}{s_j}  \right)\,.
\end{align*}

The full conditional posterior distribution of $s_j$, $j=1,\ldots,k$, is given by
\begin{align*}
    \pi(s_j|\by,\bbeta,\rhosq,&\bv,\bsigma,\lambda_1)\\
    &\propto \frac{1}{\sqrt{2\pi\rhosq s_j}} \exp\left\{ -\frac{\beta^2_j}{2\rhosq s_j} \right\} \times  \frac{\lambda^2_1}{2} \exp\left\{ -\frac{\lambda^2_1s_j}{2} \right\} \\
    &\propto s_j^{-1/2}\exp\left\{ -\frac{1}{2} \left( \frac{\beta_j^2}{\rhosq} \frac{1}{s_j} + \lambda^2_1 s_j \right) \right\}\\
    &\propto GIG\left(\frac{1}{2}, \lambda^2_1, \frac{\beta_j^2}{\rhosq}  \right)\,.
\end{align*}

The full conditional posterior distribution of $\lambda_1$ is given by
\begin{align*}
   \pi(\lambda_1|\by,\bbeta,\rhosq,&\bv,\bsigma,\bs)\\
    &\propto \prod^k_{j=1}  \frac{\lambda^2_1}{2} \exp\left\{ -\frac{\lambda^2_1s_j}{2} \right\} \\
    &\quad\quad\times \frac{b^a}{\Gamma(a)} {(\lambda^2_1)}^{a-1} \exp\left\{-b\lambda^2_1 \right\}\\
    &\propto {(\lambda^2_1)}^{a+k-1} \exp\left\{ -\left( b+\sum^k_{j=1} \frac{s_j}{2} \right)\lambda_1^2 \right\}\\
    &\propto \text{Gamma}\left( a+k,  b+\sum^k_{j=1} \frac{s_j}{2} \right)\,.
\end{align*}

\subsection{Bayesian Huberised Elastic Net Quantile regression}\label{app:gibbs-elastic}

The joint posterior distribution is as follows.
\begin{align*}
    \pi (\bbeta,\rhosq,&\bv,\bsigma,\bt,\lambda_3,\lambda_4|\by )\\
    & = \prod^n_{i=1} \frac{1}{\sqrt{8\pi\sigma_iv_i}} \exp\left\{ - \frac{(y_i-\bx_i\bbeta-(1-2\tau) v_i)^2}{8\sigma_i v_i} \right\}\\
    &\quad \quad \times \prod^n_{i=1} \frac{1}{2\rhosq K_1(\eta)} \exp\left\{ - \frac{\eta}{2} \left(\frac{\sigma_i}{\rhosq} + \frac{\rhosq}{\sigma_i} \right) \right\}\\
    &\quad\quad \times \prod^n_{i=1} \frac{\tau(1-\tau)}{2\sigma_i} \exp\left\{ -\frac{\tau(1-\tau)v_i}{2\sigma_i} \right\}\\
    &\quad\quad\times \prod^k_{j=1} \sqrt{ \frac{\lambda_4 t_j}{\pi\rhosq(t_j-1)} } \exp\left\{ -\frac{\lambda_4 t_j\beta^2_j}{\rhosq (t_j-1)} \right\} \\
    &\quad\quad \times \prod^k_{j=1} \Gamma^{-1} \left(\frac{1}{2},\tlambda \right) \sqrt{\frac{\tlambda}{t_j}} \exp\left\{-\tlambda t_j \right\} I(t_j>1)\\
    &\quad\quad\times \frac{b_1^{a_1}}{\Gamma(a_1)} {(\tlambda)}^{a_1-1} \exp\left\{-b_1\tlambda \right\}\\
    &\quad\quad\times \frac{b_2^{a_2}}{\Gamma(a_2)} \lambda_2^{a_2-1} \exp\left\{-b_2\lambda_4 \right\}\\
    &\quad\quad\times \frac{b_3^{a_3}}{\Gamma(a_3)} \eta^{a_3-1} \exp\left\{-b_3\eta \right\}\\
    &\quad\quad\times \frac{1}{\rhosq}\,.
\end{align*}

Clearly, it is obvious to see that the full conditional posterior distributions of $\sigma_i$ and $v_i$, $i=1,\ldots,n$ are the same in the Bayesian Huberised lasso quantile regression. 

The full conditional posterior distribution of $\bbeta$ is given by
\begin{align*}
    \pi (\bbeta |\by,\rhosq,&\bv,\bsigma,\bt,\lambda_3,\lambda_4)\\
    & = \prod^n_{i=1} \frac{1}{\sqrt{8\pi\sigma_iv_i}} \exp\left\{ - \frac{(y_i-\bx_i\bbeta-(1-2\tau) v_i)^2}{8\sigma_i v_i} \right\}\\
    &\quad\quad\times \prod^k_{j=1} \sqrt{ \frac{\lambda_4 t_j}{\pi\rhosq(t_j-1)} } \exp\left\{ -\frac{\lambda_4 t_j\beta^2_j}{\rhosq (t_j-1)} \right\} \\
    &\propto \exp\left\{ -\frac{1}{2} (\by-\bX\bbeta -(1-2\tau)\bv)^T\bV^{-1}(\by-\bX\bbeta -(1-2\tau)\bv) \right\}\\
    &\quad\quad\times \exp\left\{ -\frac{\lambda_4}{\rhosq} \bbeta^T\bT^{-1}\bbeta \right\} \\
    &\propto \exp\left\{ -\frac{1}{2} \left(\bbeta^T\left(\bX\bV^{-1}\bX + \frac{2\lambda_4}{\rhosq}\bT^{-1} \right) \bbeta -2\bbeta^T\bX^T\bV^{-1}(\by-(1-2\tau)\bv) \right) \right\} \\
    &\propto N\left(\bmu_{\bbeta}, \bSigma_{\bbeta} \right)\,,
\end{align*}
where $\bV=\diag(4\sigma_1v_1,\ldots,4\sigma_nv_n)$, $\bT=\diag((t_1-1)t^{-1}_1,\ldots,(t_n-1)t^{-1}_n)$, $\bSigma_{\bbeta} = \left(\bX\bV^{-1}\bX + \frac{2\lambda_4}{\rhosq}\bT^{-1} \right)^{-1}$ and $\bmu = \bSigma_{\bbeta} \bX^T\bV^{-1}(\by-(1-2\tau)\bv)$. 

The full conditional posterior distribution of $\rhosq$ is given by
\begin{align*}
    \pi(\rhosq|\by,\bbeta,&\bv,\bsigma,\bt,\lambda_3,\lambda_4)\\
    &\propto \prod^n_{i=1} \frac{1}{2\rhosq K_1(\eta)} \exp\left\{ - \frac{\eta}{2} \left(\frac{\sigma_i}{\rhosq} + \frac{\rhosq}{\sigma_i} \right) \right\}\\
    &\quad\quad\times \prod^k_{j=1} \sqrt{ \frac{\lambda_4 t_j}{\pi\rhosq(t_j-1)} } \exp\left\{ -\frac{\lambda_4 t_j\beta^2_j}{\rhosq (t_j-1)} \right\} \\
    &\quad\quad\times \frac{1}{\rhosq}\\
    &\propto {(\rhosq)}^{-n-\frac{k}{2}-1} \exp\left\{ \frac{1}{2}\left( \sum^n_{i=1}\frac{\eta}{\sigma_i} \rhosq + \left(\sum^n_{i=1}\eta\sigma_i + \sum_{j=1}^k \frac{2\lambda_4 t_j\beta^2_j}{t_j-1} \right)\frac{1}{\rhosq} \right) \right\}\\
    &\propto GIG\left( -n-\frac{k}{2}, \sum^n_{i=1}\frac{\eta}{\sigma_i}, \sum^n_{i=1}\eta\sigma_i + \sum_{j=1}^k \frac{2t_j\lambda_4\beta^2_j}{t_j-1}  \right)\,.
\end{align*}

The full conditional posterior distribution of $t_j-1$ is given by
\begin{align*}
    \pi(t_j-1|\by,\bbeta,\rhosq,&\bv,\bsigma,\lambda_3,\lambda_4)\\
    &\propto  \sqrt{ \frac{\lambda_4 t_j}{\pi\rhosq(t_j-1)} } \exp\left\{ -\frac{\lambda_4 t_j\beta^2_j}{\rhosq (t_j-1)} \right\} \\
    &\quad\quad\times \Gamma^{-1} \left(\frac{1}{2},\tlambda \right) \sqrt{\frac{\tlambda}{t_j}} \exp\left\{-\tlambda t_j \right\} I(t_j>1)\\
    &\propto (t_j-1)^{-1/2} \exp\left\{ -\frac{\lambda_4 t_j\beta^2_j}{\rhosq (t_j-1)} -\tlambda t_j   \right\}I(t_j>1)\\
    &\propto (t_j-1)^{-1/2} \exp\left\{ -\frac{1}{2}\left(\frac{2\lambda_4 \beta^2_j}{\rhosq }\frac{1}{t_j-1} + 2\tlambda(t_j-1) \right)  \right\}I(t_j-1>0)\\
    &\propto GIG\left(\frac{1}{2},2\tlambda, \frac{2\lambda_4 \beta^2_j}{\rhosq } \right)I(t_j-1>0)\,.
\end{align*}

The full conditional posterior distribution of $\tlambda$ is given by
\begin{align*}
    \pi(\tlambda|\by,\bbeta,\rhosq,&\bv,\bsigma,\bt,\lambda_3,\lambda_4)\\
    &\propto \prod^k_{j=1} \Gamma^{-1} \left(\frac{1}{2},\tlambda \right) \sqrt{\frac{\tlambda}{t_j}} \exp\left\{-\tlambda t_j \right\} I(t_j>1)\\
    &\quad\quad\times \frac{b_1^{a_1}}{\Gamma(a_1)} {(\tlambda)}^{a_1-1} \exp\left\{-b_1\tlambda \right\}\\
    &\propto \Gamma^{-k} \left(\frac{1}{2},\tlambda \right)   {(\tlambda)}^{ \frac{k}{2}+a_1-1} \exp\left\{-\left(\sum^k_{j=1} t_j +b_1 \right) \tlambda \right\}\,.
\end{align*}
As it is infeasible to directly sample from $\pi(\tlambda|\by,\bbeta,\rhosq,\bv,\bsigma,\bt,\lambda_3,\lambda_4)$, the one-step Metropolis-Hastings algorithm is employed. Following \cite{LiEtAl2010}, the proposal distribution is \\$q(\tlambda|\bt)\sim \text{Gamma}\left(k+a_1, b_1+\sum^k_{j=1}(t_j-1) \right)$. They showed that 
\begin{align*}
    \underset{\tlambda\rightarrow \infty}{\lim}\ \frac{\sqrt{\tlambda}\exp(\tlambda)}{\Gamma^{-1}\left( \frac{1}{2},\tlambda\right)}  = 1\,,
\end{align*}
implies that
\begin{align*}
    \underset{\tlambda\rightarrow \infty}{\lim}\ \frac{\pi(\tlambda|\by,\bbeta,\rhosq,\bv,\bsigma,\bt,\lambda_3,\lambda_4)}{q(\tlambda|\bt)}\,, 
\end{align*}
exists and equals to some positive constant. Hence, the tail behaviours of $q(\tlambda|\bt)$ and $\pi(\tlambda|\by,\bbeta,\rhosq,\bv,\bsigma,\bt,\lambda_3,\lambda_4)$ are similar.

The full conditional posterior distribution of $\lambda_4$ is given by
\begin{align*}
    \pi(\lambda_4|\by,\bbeta,\rhosq,&\bv,\bsigma,\bt,\lambda_3)\\
    &\propto \prod^k_{j=1} \sqrt{ \frac{\lambda_4 t_j}{\pi\rhosq(t_j-1)} } \exp\left\{ -\frac{\lambda_4 t_j\beta^2_j}{\rhosq (t_j-1)} \right\} \\
    &\quad\quad\times \frac{b_2^{a_2}}{\Gamma(a_2)} \lambda_2^{a_2-1} \exp\left\{-b_2\lambda_4 \right\}\\ 
    &\propto \lambda_4^{\frac{k}{2}+a_2-1} \exp\left\{ -\left( \sum^k_{j=1}   \frac{ t_j\beta^2_j}{\rhosq (t_j-1)} + b_2 \right)\lambda_4 \right\} \\
    &\propto \text{Gamma}\left(\frac{k}{2}+a_2, \sum^k_{j=1}   \frac{ t_j\beta^2_j}{\rhosq (t_j-1)} + b_2 \right)\,.
\end{align*}

\section{Results for Simulation Studies and Real Data Examples}\label{app:figures}

\begin{figure}[H]
\centering
\includegraphics[width=\textwidth]{HBQR_BL_RMSE_tau25.pdf}
\includegraphics[width=\textwidth]{HBQR_BEN_RMSE_tau25.pdf}
\includegraphics[width=\textwidth]{BQR_BL_RMSE_tau25.pdf}
\includegraphics[width=\textwidth]{BQR_BEN_RMSE_tau25.pdf}
\caption{Boxplots of RMSE based on 300 replications in six simulation scenarios for HBQR-BL, HBQR-EN,  BQR-BL and BQR-EN in this order ($\tau=0.25$).}
\label{fig:boxplot-rmse25}
\end{figure}
\FloatBarrier

\begin{figure}[H]
\centering
\includegraphics[width=\textwidth]{HBQR_BL_MMAD_tau25.pdf}
\includegraphics[width=\textwidth]{HBQR_BEN_MMAD_tau25.pdf}
\includegraphics[width=\textwidth]{BQR_BL_MMAD_tau25.pdf}
\includegraphics[width=\textwidth]{BQR_BEN_MMAD_tau25.pdf}
\caption{Boxplots of MMAD based on 300 replications in six simulation scenarios for HBQR-BL, HBQR-EN,  BQR-BL and BQR-EN in this order ($\tau=0.25$).}
\label{fig:boxplot-mmad25}
\end{figure}
\FloatBarrier

\begin{figure}[H]
\centering
\includegraphics[width=\textwidth]{HBQR_BL_AL_tau25.pdf}
\includegraphics[width=\textwidth]{HBQR_BEN_AL_tau25.pdf}
\includegraphics[width=\textwidth]{BQR_BL_AL_tau25.pdf}
\includegraphics[width=\textwidth]{BQR_BEN_AL_tau25.pdf}
\caption{Boxplots of AL based on 300 replications in six simulation scenarios for HBQR-BL, HBQR-EN,  BQR-BL and BQR-EN in this order ($\tau=0.25$).}
\label{fig:boxplot-al25}
\end{figure}
\FloatBarrier

\begin{figure}[H]
\centering
\includegraphics[width=\textwidth]{HBQR_BL_RMSE_tau75.pdf}
\includegraphics[width=\textwidth]{HBQR_BEN_RMSE_tau75.pdf}
\includegraphics[width=\textwidth]{BQR_BL_RMSE_tau75.pdf}
\includegraphics[width=\textwidth]{BQR_BEN_RMSE_tau75.pdf}
\caption{Boxplots of RMSE based on 300 replications in six simulation scenarios for HBQR-BL, HBQR-EN,  BQR-BL and BQR-EN in this order ($\tau=0.75$).}
\label{fig:boxplot-rmse75}
\end{figure}
\FloatBarrier

\begin{figure}[H]
\centering
\includegraphics[width=\textwidth]{HBQR_BL_MMAD_tau75.pdf}
\includegraphics[width=\textwidth]{HBQR_BEN_MMAD_tau75.pdf}
\includegraphics[width=\textwidth]{BQR_BL_MMAD_tau75.pdf}
\includegraphics[width=\textwidth]{BQR_BEN_MMAD_tau75.pdf}
\caption{Boxplots of MMAD based on 300 replications in six simulation scenarios for HBQR-BL, HBQR-EN,  BQR-BL and BQR-EN in this order ($\tau=0.75$).}
\label{fig:boxplot-mmad75}
\end{figure}
\FloatBarrier

\begin{figure}[H]
\centering
\includegraphics[width=\textwidth]{HBQR_BL_AL_tau75.pdf}
\includegraphics[width=\textwidth]{HBQR_BEN_AL_tau75.pdf}
\includegraphics[width=\textwidth]{BQR_BL_AL_tau75.pdf}
\includegraphics[width=\textwidth]{BQR_BEN_AL_tau75.pdf}
\caption{Boxplots of AL based on 300 replications in six simulation scenarios for HBQR-BL, HBQR-EN,  BQR-BL and BQR-EN in this order ($\tau=0.75$).}
\label{fig:boxplot-al75}
\end{figure}
\FloatBarrier

\begin{figure}[H]
    \centering
    \includegraphics[width=\textwidth]{HBQR_BL_Eta_tau25.pdf}
    \includegraphics[width=\textwidth]{HBQR_BEN_Eta_tau25.pdf}
    \caption{Boxplots of posterior median of $\eta$ based on 300 replications in six simulation scenarios for HBQR-BL (top) and HBQR-EN (bottom) ($\tau=0.25$).}
    \label{fig:boxplot-eta25}
\end{figure}
\FloatBarrier

\begin{figure}[H]
    \centering
    \includegraphics[width=\textwidth]{HBQR_BL_Eta_tau75.pdf}
    \includegraphics[width=\textwidth]{HBQR_BEN_Eta_tau75.pdf}
    \caption{Boxplots of posterior median of $\eta$ based on 300 replications in six simulation scenarios for HBQR-BL (top) and HBQR-EN (bottom) ($\tau=0.75$).}
    \label{fig:boxplot-eta75}
\end{figure}
\FloatBarrier

\begin{figure}[H]
    \centering
    \includegraphics[width=0.9\textwidth]{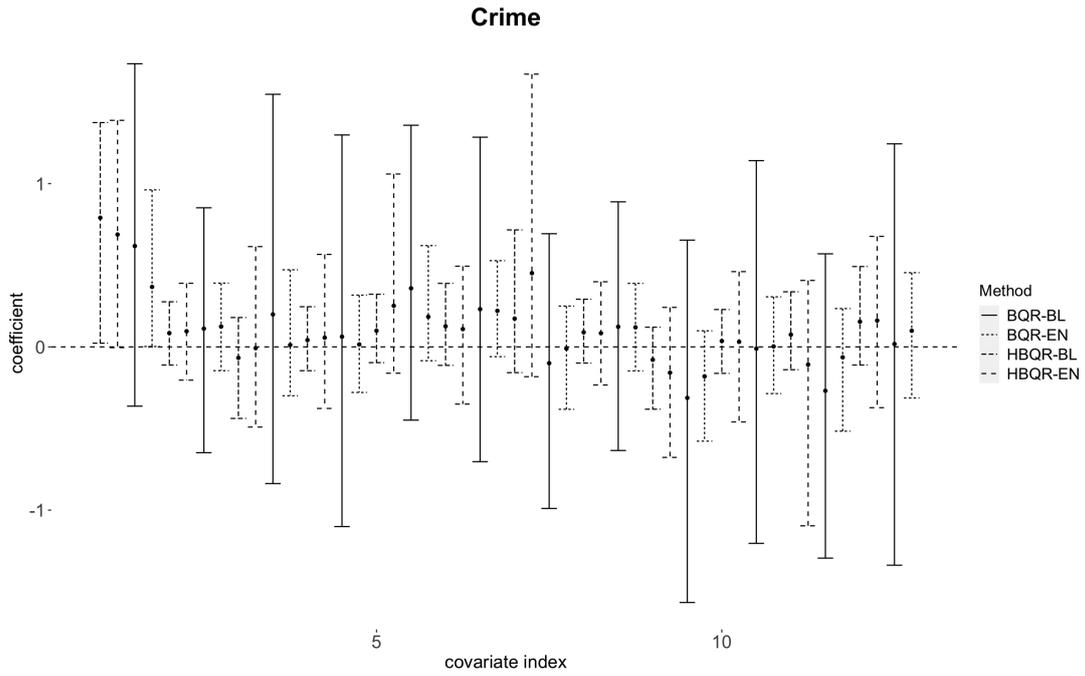}
    \caption{Posterior medians and 95\% credible intervals of the regression coefficients at $\tau=0.1$ in the Bayesian quantile regression with Bayesian lasso (BQR-BL), Bayesian quantile regression with elastic net (BQR-EN) and the proposed Bayesian quantile regression with Bayesian lasso (HBQR-BL) and elastic net (HBQR-EN), applied to the Crime data.}
    \label{fig:crimeCI1}
\end{figure}
\FloatBarrier

\begin{figure}[H]
    \centering
    \includegraphics[width=0.9\textwidth]{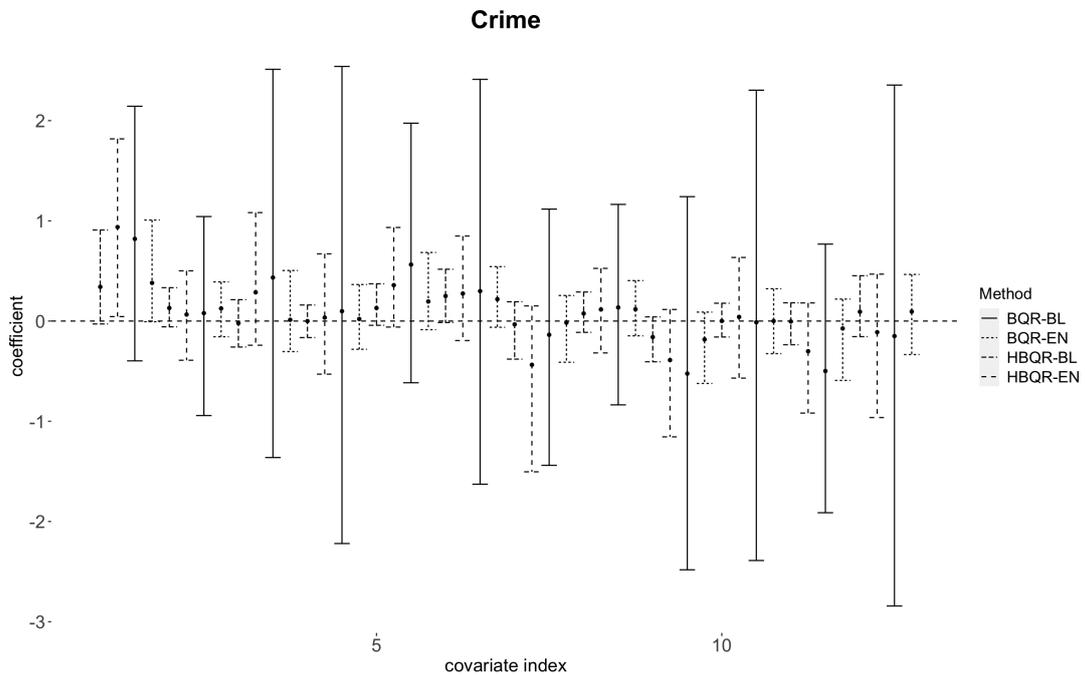}
    \caption{Posterior medians and 95\% credible intervals of the regression coefficients at $\tau=0.9$ in the Bayesian quantile regression with Bayesian lasso (BQR-BL), Bayesian quantile regression with elastic net (BQR-EN) and the proposed Bayesian quantile regression with Bayesian lasso (HBQR-BL) and elastic net (HBQR-EN), applied to the Crime data.}
    \label{fig:crimeCI9}
\end{figure}
\FloatBarrier

\begin{figure}[H]
    \centering
    \includegraphics[width=0.9\textwidth]{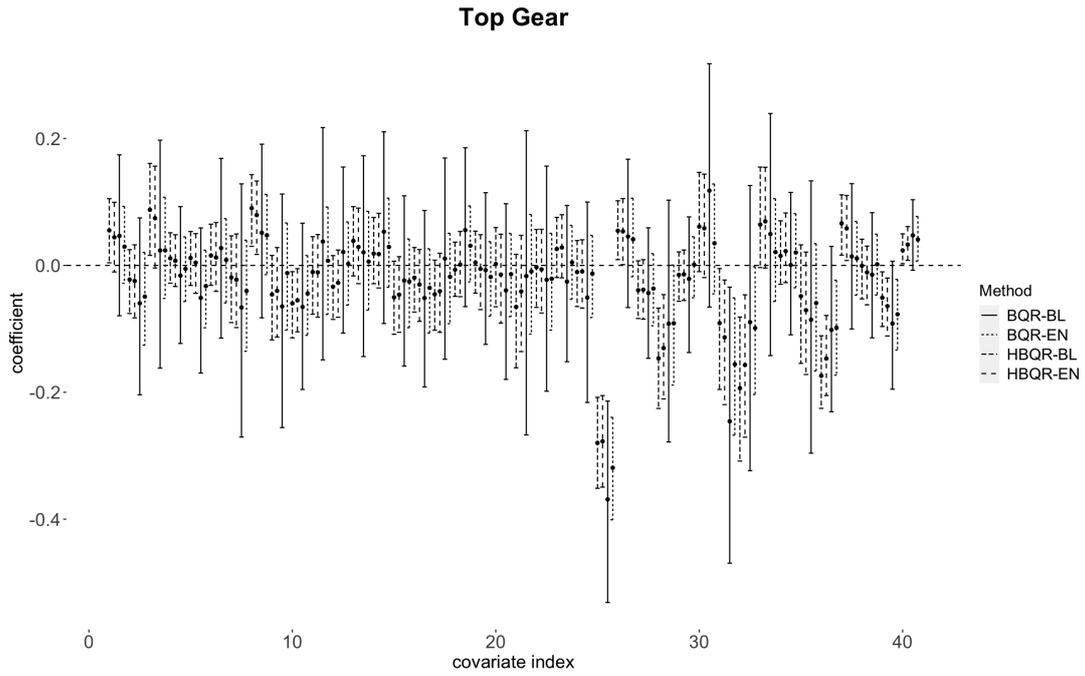}
    \caption{Posterior medians and 95\% credible intervals of the regression coefficients at $\tau=0.1$ in the Bayesian quantile regression with Bayesian lasso (BQR-BL), Bayesian quantile regression with elastic net (BQR-EN) and the proposed Bayesian quantile regression with Bayesian lasso (HBQR-BL) and elastic net (HBQR-EN), applied to the Top Gear data.}
    \label{fig:topgearCI1}
\end{figure}
\FloatBarrier

\begin{figure}[H]
    \centering
    \includegraphics[width=0.9\textwidth]{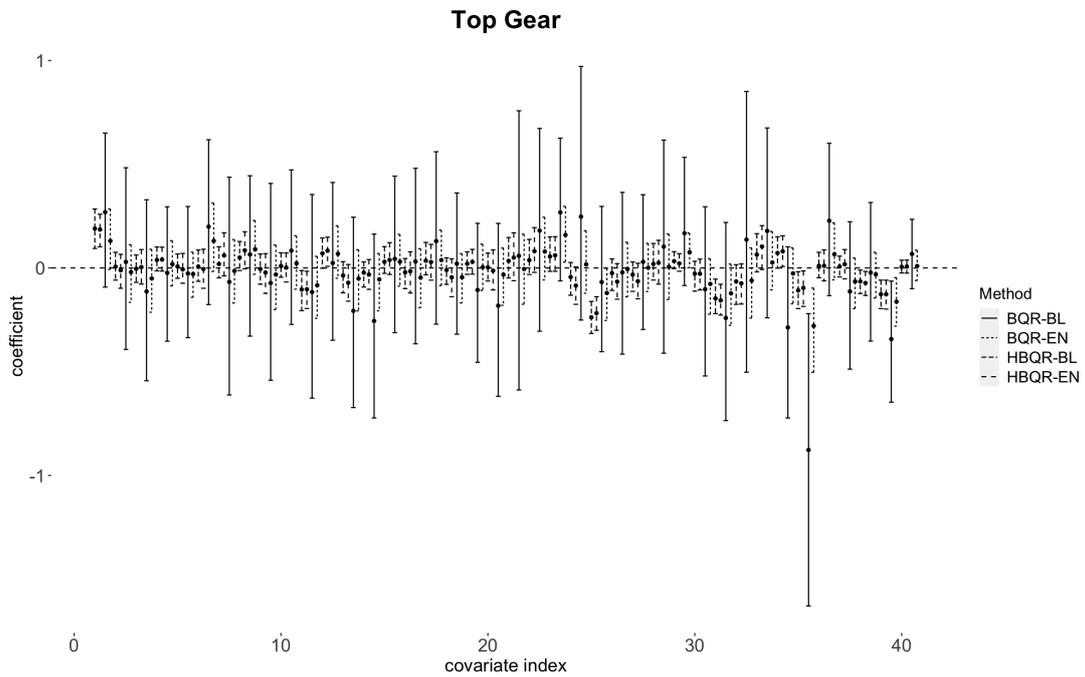}
    \caption{Posterior medians and 95\% credible intervals of the regression coefficients at $\tau=0.9$ in the Bayesian quantile regression with Bayesian lasso (BQR-BL), Bayesian quantile regression with elastic net (BQR-EN) and the proposed Bayesian quantile regression with Bayesian lasso (HBQR-BL) and elastic net (HBQR-EN), applied to the Top Gear data.}
    \label{fig:topgearCI9}
\end{figure}
\FloatBarrier

\end{document}